\documentclass[onefignum,onetabnum,final]{siamart171218}

\newcommand{\RR}{\mathbb{R}}

\newcommand{\ZZ}{\mathbb{Z}}

\def\RR{\mathbb{R}}
\def\S{S}

\def\c{\boldsymbol{c}}
\def\x{\boldsymbol{x}}

\def\p{\boldsymbol{p}}

\def\r{\boldsymbol{r}}

\def\f{\boldsymbol{f}}

\def\bphi{\boldsymbol{\phi}}
\def\bpsi{\boldsymbol{\psi}}
\def\bxi{\boldsymbol{\xi}}
\def\0{\mathbf{0}}
\def\1{\mathbf{1}}
\def\B{B}

\def\kk{\kappa}
\def\bk{{B_\kk}}
\def\Bk{{B_\kk}}
\def\bkT{{B_\kk^\transp}}

\def\jmax{j_\mathrm{max}}
\def\ksim{\>\> \text{\stackunder[0pt]{$\sim$}{$\scriptscriptstyle\kk$}} \>\>}

\newcommand{\transp}{{\scriptstyle{\mathsf{T}}}}
\newcommand{\tildec}[1]{\, \tilde{\! #1}}

\DeclareMathOperator{\bd}{bd}
\DeclareMathOperator{\hl}{hl}
\DeclareMathOperator{\nat}{nat}

\DeclareMathOperator{\sign}{sign}
\DeclareMathOperator{\supp}{supp}
\DeclareMathOperator{\diag}{diag}

\DeclareMathOperator{\Ccut}{Ccut}
\DeclareMathOperator{\Cvol}{Cvol}
\DeclareMathOperator{\Kcut}{\kappa Cut}

\DeclareMathOperator{\SignedRatioCut}{SignedRatioCut}
\DeclareMathOperator{\SignedNormalizedCut}{SignedNormalizedCut}

\graphicspath{{./figs/}}

\usepackage{caption}
\usepackage{subcaption}
\usepackage{mathtools, txfonts}
\usepackage{fourier}
\usepackage{version, cancel}
\usepackage[algo2e,ruled,vlined]{algorithm2e}
\usepackage{stackengine}
\excludeversion{hide}

\title{Multiscale Transforms for Signals on Simplicial Complexes}

\author{Naoki Saito\thanks{Department of Mathematics, University of California, Davis
  (\email{saito@math.ucdavis.edu}, ).}
\and Stefan C. Schonsheck \thanks{Department of Mathematics, University of California, Davis
  (\email{scschonsheck@ucdavis.edu}).}
\and Eugene Shvarts \thanks{Department of Mathematics, University of California, Davis
  (\email{eshvarts@ucdavis.edu})} }

\begin{document}

\maketitle


\begin{abstract}
Our previous multiscale graph basis dictionaries/graph signal transforms---Generalized Haar-Walsh Transform (GHWT); Hierarchical Graph Laplacian Eigen Transform (HGLET); Natural Graph Wavelet Packets (NGWPs); and their relatives---were developed for analyzing data recorded on nodes of a given graph. In this article, we propose their generalization for analyzing data recorded on edges, faces (i.e., triangles), or more generally $\kappa$-dimensional simplices of a simplicial complex (e.g., a triangle mesh of a manifold). The key idea is to use the Hodge Laplacians and their variants for hierarchical partitioning of a set of $\kappa$-dimensional simplices in a given simplicial complex, and then build localized basis functions on these partitioned subsets. We demonstrate their usefulness for data representation on both illustrative synthetic examples and real-world simplicial complexes generated from a co-authorship/citation dataset and an ocean current/flow dataset.
\end{abstract}


\section{Introduction}

For conventional digital signals and images sampled on regular lattices,
\emph{multiscale basis dictionaries}, i.e., \emph{wavelet packet dictionaries}
including \emph{wavelet bases}, \emph{local cosine dictionaries}, and their
variants (see, e.g., \cite[Chap.~4, 7]{WICK-WPK}, \cite[Chap.~6, 7]{Meyer-AppBook2}, \cite[Chap.~8]{MALLAT-BOOK3}), have a proven track record of success: 
\emph{JPEG 2000} Image Compression Standard~\cite[Sec.~15.9]{SAYOOD3}; 
\emph{Modified Discrete Cosine Transform} (MDCT) in MP3~\cite[Sec.~16.3]{SAYOOD3};
\emph{discriminant feature extraction} for signal classification~\cite{SAITO-COIF-JMIV, SAITO-COIF-SONIC, SAITO-COIF-GESHWIND-WARNER}, just to name a few.
Considering the abundance of data measured on graphs and networks and
the increasing importance to analyze such data (see, e.g., \cite{EASLEY-KLEINBERG, NEWMAN2, CHUNG-LU, LOVASZ-BOOK, SHUMAN-ETAL}), 
it is quite natural to lift/generalize these dictionaries to the graph setting.
Our group have developed the graph versions of the block/local cosine and
wavelet packet dictionaries for analysis of graph signals \emph{sampled at
nodes} so far. These include the \emph{Generalized Haar-Walsh Transform} (GHWT)~\cite{irion2014generalized}, the \emph{Hierarchical Graph Laplacian Eigen Transform} (HGLET)~\cite{irion2014hierarchical}, the \emph{Natural Graph Wavelet Packets} (NGWPs)~\cite{CLONINGER-LI-SAITO},
and their relatives~\cite{IRION-SAITO-MLSP16, shao2019extended, saito2022eghwt};
see also \cite{IRION-SAITO-SPIE, IRION-SAITO-TSIPN}. Some of these will be
briefly reviewed in the later sections.

In this article, we propose their generalization for analyzing data recorded on
edges, faces (i.e., triangles), or more generally cells (i.e., polytopes) of
a class of special graphs called \emph{simplicial complexes} (e.g., a triangle
mesh of a manifold).
The key idea is to use the \emph{Hodge Laplacians} and their variants for
hierarchical partitioning of a set of $\kappa$-dimensional simplices in a given
simplicial complex, and then build localized basis functions on these partitioned
subsets.
We demonstrate their usefulness for data representation on both illustrative synthetic examples and real-world simplicial complexes generated from a co-authorship/citation dataset and an ocean current/flow dataset.

\subsection{Related work}

Graph-based methods for analyzing data have been widely adopted in many domains, \cite{bruna2013spectral, ortega2018graph, dong2020graph}. Often, these graphs are fully defined by data (such as a graph of social media ``friends"), but they can also be induced through the persistence homology of generic point clouds \cite{carlsson2009topology}. In either case, the vast majority of these analytical techniques deal with signals which are defined on the vertices (or nodes) of a given graph. More recently, there has been a surge in interest in studying signals defined on edges, triangles, and higher-dimensional substructures within the graph \cite{carlsson2009topology, shuman2013emerging, giusti2016two, barbarossa2020topological, chen2021helmholtzian}. The fundamental tool employed for analyzing these signals, the \emph{Hodge Laplacian}, has been studied in the context of differential geometry for over half a century but has only recently entered the toolbox of applied mathematics. This rise in popularity is largely due to the adaptation of discrete differential geometry \cite{Crane:2013:DGP} in applications in computer vision \cite{lim2020hodge, roddenberry2022signal}, statistics \cite{jiang2011statistical}, topological data analysis \cite{chen2021helmholtzian, schonsheck2022spherical}, and network analysis \cite{schaub2020random}.

One of the key challenges to applying wavelets and similar constructions to vertex-based graph signals is that graphs lack a natural translation operator, which prevents the construction of convolutional operators and traditional Littlewood-Paley theory \cite{IRION-SAITO-SPIE, kondor2018generalization, schonsheck2022parallel}. This challenge is also present for general $\kk$-dimensional simplices. One method for overcoming this difficulty is to perform convolution solely in the ``frequency'' domain and define wavelet-like bases entirely in the coefficient space of the Laplacian (or in this case Hodge Laplacian) transform. Following this line of research, there have been several approaches to defining wavelets \cite{roddenberry2022hodgelets} and convolutional neural networks \cite{ebli2020simplicial} in which the input signal is transformed in a series of coefficients in the eigenspace of the Hodge Laplacian. Unfortunately, the atoms (or basis vectors) generated by these methods are not always locally supported, and it can be difficult to interpret their role in analyzing a given graph signal. 

An alternative path to the creation of wavelet-like dictionaries and transforms is to first develop a hierarchical block decomposition of the domain and then use this to develop multiscale transforms \cite{irion2014hierarchical, irion2014generalized, saito2022eghwt}. These techniques rely on recursively computing bipartitions of the domain and then generating localized bases on the subsets of the domain. In this work, we propose a simplex analog to the Fiedler vector \cite{holzrichter1999graph} to solve a relaxed version of the simplex-normalized-cut problem, which we can apply iteratively to develop a hierarchical bipartition of the $\kk$-dimensional simplices in a simplicial complex. From here, we are able to apply the general scheme of  \cite{irion2014hierarchical} and \cite{irion2014generalized} to develop the
\emph{Hierarchical Graph Laplacian Eigen Transform} and the \emph{Generalized Haar-Walsh Transform}, respectively, for a given collection of simplices of an arbitrarily high order. As a result, we can also generate orthonormal Haar bases, orthonormal Walsh bases, as well as data-adaptive orthonormal bases using the best-basis selection method~\cite{coifman1992entropy}. 

The main challenge in lifting these transforms to the simplicial setting lies in the simplex orientations, which cause the resulting Laplacians generally to contain mixed positive and negative off-diagonal elements. We are no longer guaranteed a non-negative Perron vector \cite{bapat1997nonnegative} to use as a DC component, and so must incorporate the orientation information both in order to develop a Fiedler vector appropriate for partitioning the $\kk$-dimensional simplices of a complex, and for interpreting the nature of the resulting partition. Further challenges lie in there being multiple ways to define adjacency between simplices, multiple ways to generalize simplex weights of pairs of adjacent simplices and multiple ways to balance the ``upper" and ``lower" parts of the Hodge Laplacian. 

\subsection{Outline}
This article is organized as follows: In Section~\ref{sec:simplices} we formally describe simplicial complexes and how their geometry leads to notions of adjacency and orientation. This allows us to define discrete differential operators acting on signals defined on the complex, which in turn are constructed from \emph{boundary operators} that map between the $\kk$ and $\kk \pm 1$ degree faces of the complex.
In Section~\ref{sec:hodgelap} we use these boundary operators to describe the Hodge Laplacian and discuss several different variants, some analogous to different normalizations of the graph Laplacian and some more novel. In Section~\ref{sec:Fiedler} we show how the eigenvectors of the Hodge Laplacian can be use to solve relaxed-cut-like problems to partition a complex. We also develop \emph{hierarchical bipartitions}, which decompose a given complex roughly in half at each level until we are left with a division into individual elements. In Section~\ref{sec:Haar} we use these bipartitions to develop orthonormal Haar bases. In Section~\ref{sec:dict}, we create overcomplete dictionaries based on given bipartitions and, as a consequence, are also able to define a canonical orthonormal Walsh basis. At the end of this section we state two theorems which bound the decay rate of the dictionary coefficients and approximation power of our dictionaries. In Section~\ref{sec:numexp}, we present numerical experiments on both illustrative synthetic examples and real-world problems in signal approximation, clustering, and supervised classification. Finally, we conclude this article with Section~\ref{sec:concl} discussing potential future work.

We have implemented our multiscale simplicial signal transforms in Julia and Python, and code which builds the corresponding basis dictionaries, and was used to generate the figures in this article, is available at:\\ \texttt{https://github.com/UCD4IDS/MultiscaleSimplexSignalTransforms.jl}.


\section{Simplicial Complexes}
\label{sec:simplices}

In this section we review concepts from algebraic topology to formally define simplicial complexes and introduce some notions of how two simplices can be ``adjacent.'' For a more thorough review, see \cite{carlsson2009topology, giusti2016two}. Given a vertex set $V = \{v_1, \ldots, v_n\}$, a \emph{$\kk$-simplex} $\sigma$ is a $(\kk+1)$-subset of $V$. 
A \emph{face} of $\sigma$ is a $\kk$-subset of $\sigma$, and so $\sigma$ has $\kk+1$ faces. A \emph{co-face} of $\sigma$ is a $(\kk+1)$-simplex, of which $\sigma$ is a face.

Suppose $\sigma = \{v_{i_1}, \ldots, v_{i_{\kk+1}}\}$, $i_1 < \cdots < i_{\kk+1}$, and $\alpha \subset \sigma$ is its face. Then, $\sigma \setminus \alpha$ consists of a single vertex; let $v_{i_{\ell^*}}$ be that vertex where $1 \leq \ell^* \leq \kk+1$.
Then the \emph{natural parity} of $\sigma$ with respect to its face $\alpha$ is defined as
$$
\nat(\sigma, \alpha) := (-1)^{\ell^*+1}~~.
$$
When $\alpha$ is not a face of $\sigma$, $\nat(\sigma, \alpha) = 0$.
The natural parity of $\kk$-simplices with respect to their faces generalizes the idea of a directed edge having a head vertex and a tail vertex, and is ``natural'' because it disallows situations analogous to a directed edge with two heads or two tails.


A \emph{simplicial complex} $C$ is a collection of simplices closed under subsets, where if $\sigma \in C$, then $\alpha \subset \sigma \implies \alpha \in C$. 
In particular, if $\sigma \in C$, so does each face of $\sigma$.
Let $\kk_{\mathrm{max}}(C) \coloneqq \max\left\{\kk \, \vert \, \sigma\in C \text{ is a $\kk$-simplex}\right\}$, and let $C_\kk$ denote the set of $\kk$-simplices in $C$ for each $\kk = 1,\ldots, \kk_{\mathrm{max}}$.
When $\kk > \kk_{\mathrm{max}}$, $C_\kk = \emptyset$.
We also refer to $C$ as a \emph{$\kk$-complex} to note that $\kk_{\mathrm{max}}(C) = \kk$.
Let a \emph{$\kk$-region} of $C$ refer to any non-empty subset of $C_\kk$.

Let $C$ be a simplicial complex, and $\sigma, \tau \in C_\kk$, for some $\kk>0$.
When $\sigma, \tau$ share a face, they are \emph{weakly adjacent}, denoted by $\sigma \sim \tau$.
Their shared boundary face is denoted $\bd(\sigma, \tau)$. 
When $\sigma \sim \tau$, additionally they both share a co-face, their \emph{hull}, denoted by $\hl(\sigma, \tau)$.
If $\sigma, \tau \in C$, $\sigma \sim \tau$, and $\hl(\sigma, \tau) \in C$, then $\sigma, \tau$ are \emph{strongly adjacent}, denoted by $\sigma \simeq \tau$.
If $\sigma \sim \tau$, but $\sigma \nsimeq \tau$ in $C$, then $\sigma, \tau$ are \emph{$\kk$-adjacent}, denoted $\sigma \ksim \tau$. 

\begin{figure}
    \centering
    \includegraphics[width=.4\textwidth]{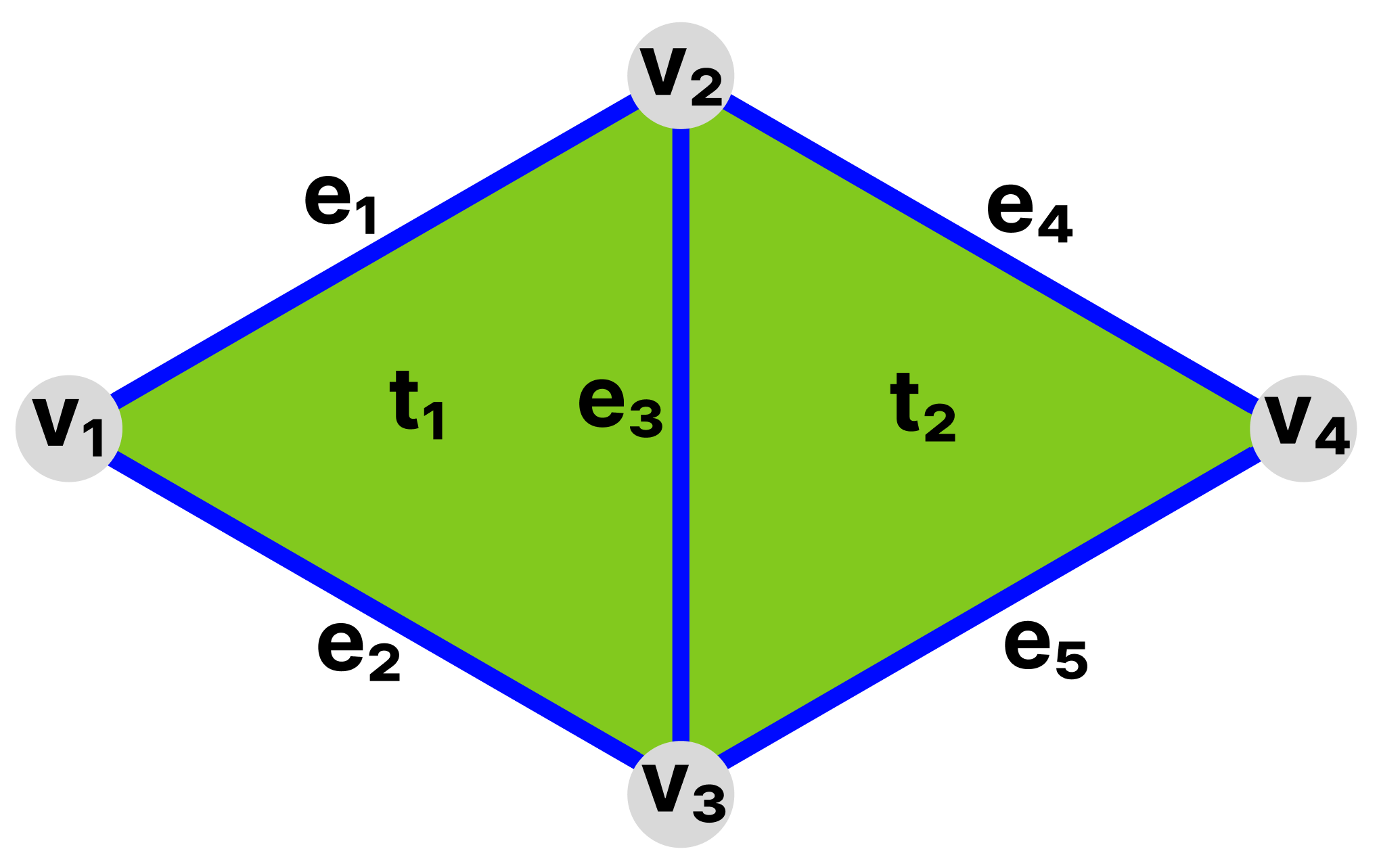}
    \caption{In this small $2$-complex $C$, $e_1 \sim e_4$ because they share the face $v_2$, and $e_1 \sim e_2$ because they share the face $v_1$. Further $e_1 \simeq e_2$ because their hull $t_1 \in C$, but $e_1 \nsimeq e_4$, so that $e_1 \>\> \text{\stackunder[0pt]{$\sim$}{$\scriptscriptstyle1$}} \>\> e_4$. We have $t_1 \sim t_2$ because they share the face $e_3$, but as $\hl(t_1, t_2) \notin C$, we have $t_1 \>\> \text{\stackunder[0pt]{$\sim$}{$\scriptscriptstyle2$}}\>\> t_2$.}
    \label{fig:twotriangle}
\end{figure}

\subsection{Oriented Simplicial Complexes and Boundary Operators}
An \emph{oriented} simplex $\sigma$ further has an orientation $p_\sigma \in \{\pm1\}$, which indicates whether its parity with its faces is the same as, or opposite to, its natural parity.
When $p_\sigma = +1$, we say $\sigma$ is in \emph{natural orientation}.
For example, a directed edge $e = (v_i,v_j)$ for $i<j$ is in natural orientation, while if $i>j$, $p_e = -1$.
An oriented simplicial complex contains at most one orientation for any given simplex.

Let $X_\kk$ be the space of real-valued functions on $C_\kk$ for each $\kk \in \{0, 1, \ldots, \kk_{\mathrm{max}}(C)\}$.
In the case of graphs, $X_0$ consists of functions taking values on vertices, or graph signals. 
$X_1$ consists of functions on edges, or edge flows. 
A function in $X_1$ is positive when the corresponding flow direction agrees with the edge orientation, and negative when the flow disagrees.
$X_2$ consists of functions on oriented triangles. 

Given an oriented simplicial complex $C$, for each $\kk \in \{0, 1, \ldots, \kk_\mathrm{max} \}$,
the \emph{boundary operator} is a linear operator $\bk: X_{\kk+1} \mapsto X_\kk$, where for $\sigma \in C_{\kk+1}$, $\alpha \in C_\kk$, the corresponding matrix entries are $[\Bk]_{\alpha\sigma} = p_\sigma p_\alpha \nat(\sigma, \alpha)$. 
Likewise, the \emph{coboundary operator} for each $\kk\in \{0, 1, \ldots, \kk_\mathrm{max}\}$ is just $\bkT: X_\kk \rightarrow X_{\kk+1}$, the adjoint to $\bk$. 
The expression of the entries of $\bk$ as the \emph{relative} orientation between simplex and face suggests that these are a natural way to construct functions taking local signed averages, according to adjacency in the simplicial complex.



\subsection{Data on Simplicial Complexes}

Signal processing on simplicial complexes arises as a natural problem in the setting where richer structure is incorporated in data, than just scalar functions and pairwise relationships.
In this article, we assume the input data is given on an existing simplicial complex.

A simple directed graph $G = (V, E)$ can always be represented as an oriented $1$-complex 
$\tilde G$, with each directed edge $e=(v_i, v_j)$ inserted as a $1$-simplex having orientation $p_e = \sign(j-i)$.
With this convention, natural orientation corresponds to the agreement of the edge direction with the global ordering of the vertices.

\section{Hodge Laplacian}
\label{sec:hodgelap}
The boundary operators just introduced represent \emph{discrete differential operators} encoding the structure of $\kk$-regions in a simplicial complex, and so can be building blocks towards a spectral analysis of functions on those regions. 
For analyzing functions on $\kk$-simplices with $\kk>0$, we will construct operators based on the \emph{Hodge Laplacian}, or \emph{$\kk$-Laplacian}. 
As in \cite{lim2020hodge}, the \emph{combinatorial} $\kk$-Laplacian is defined for $\kk$-simplices as
$$
L_\kk \coloneqq \B_{\kk-1}^\transp \B_{\kk-1} + \B_\kk \B_\kk^\transp~~.
$$
We refer to $L^{\vee}_\kk \coloneqq \B_{\kk-1}^\transp \B_{\kk-1}$ and $L^\wedge_\kk \coloneqq \B_\kk \B_\kk^\transp$ as the \emph{lower} and \emph{upper} $\kk$-Laplacians, respectively.


\subsection{Simplex consistency}

Let $C$ be an oriented simplicial complex, and $\sigma\sim\tau \in C_\kk$, with $\alpha = \bd(\sigma, \tau)$. Then we may write $L_\kk$ as $\diag(L_\kk) - S_\kk$, where for $\kk>0$, $\S_\kk$ is the \emph{signed adjacency matrix}
$$
\left[\S_\kk\right]_{\sigma\tau} \coloneqq \begin{cases}
-p_\sigma p_\tau \nat(\sigma, \alpha)\nat(\tau, \alpha) & \sigma \ksim \tau\\
0 & \text{otherwise}
\end{cases}~~.
$$
When $S_\kk > 0$, we say $\sigma, \tau$ are \emph{consistent}, and otherwise they are \emph{inconsistent}.
A consistent pair of simplices view their shared boundary face in opposite ways; one as a head face, and the other as a tail face.
An inconsistent pair of simplices view their shared boundary face identically.
In the case of $\kk=1$, two directed edges are consistent when they \emph{flow} into each other at their boundary vertex, and are inconsistent when they \emph{collide} at the boundary vertex, either both pointing toward it, or both pointing away. Cases for $\kk=1, 2$ are demonstrated in Figure \ref{fig:consistency}. 

\begin{figure}
  \centering\includegraphics[width=.45\textwidth]{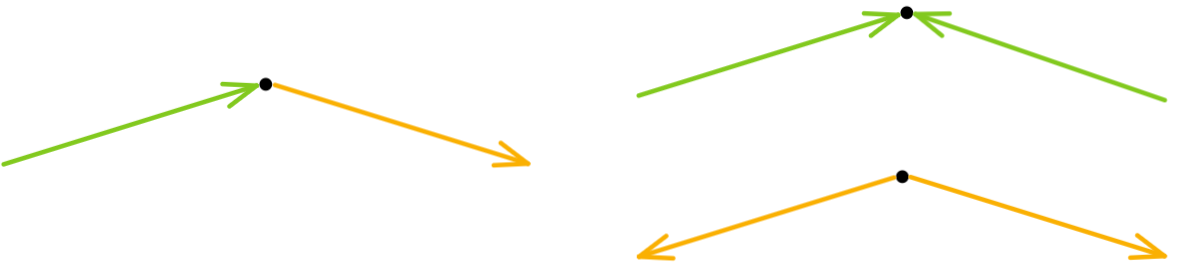}
  \hspace{.05\textwidth}
  \centering\includegraphics[width=.45\textwidth]{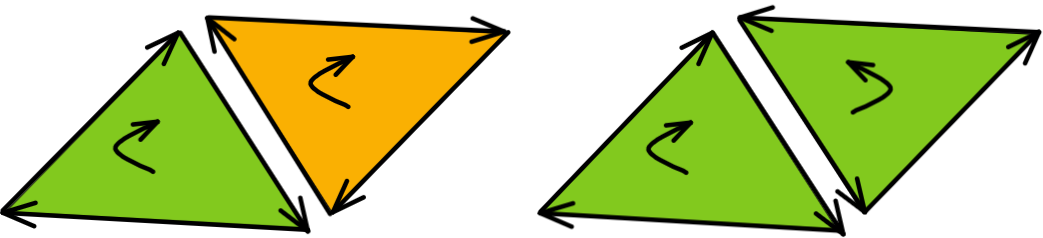}
  \caption{Pairs of $\kk$-simplices demonstrating consistency at their boundary face, for $\kk=1,2$. The mixed-color pairs are consistent, and the same-color pairs are inconsistent.}
  \label{fig:consistency}
\end{figure}
The combinatorial $\kk$-Laplacian represents signed adjacency between $\kk$-adjacent simplices via their consistency.
In particular, this means that $L_\kk$ depends only on the orientations of simplices in $C_\kk$.
Naively, constructing the boundary matrices $\B_{\kk-1}, \B_\kk$ then additionally requires superfluous sign information -- the orientation of each member of both $C_{\kk-1}$ and $C_{\kk+1}$.
This situation exactly mirrors that of the graph Laplacian $L_0$:
in order to construct $L_0$ for an undirected graph via the product $B_0 B_0^\transp$, one must assign an arbitrary direction to each edge, and the resulting Laplacian is independent of that choice of directions.

\subsection{Weighted and Normalized Hodge Laplacians}

In order to introduce a \emph{weighted} simplicial complex, consider the symmetrically normalized graph Laplacian
$$
L_0^\mathrm{sym} \coloneqq D_0^{-1/2}B_0 D_1 B_0^\transp D_0^{-1/2} = \left(D_0^{-1/2} B_0 D_1^{1/2}\right) \left(D_0^{-1/2} B_0 D_1^{1/2}\right)^\transp~~,
$$
where $D_0 = \diag(\vert B_0 \vert D_1 \1)$, the diagonal matrix of (weighted) vertex degrees, and $D_1$ is the diagonal matrix of edge weights. Letting $D_\kk$ generally refer to a diagonal matrix containing $\kk$-simplex weights, we proceed as in \cite{chen2021helmholtzian} and define the \emph{weighted symmetrically normalized $\kk$-Laplacian} as
$$
L_\kk^\mathrm{sym} \coloneqq \mathfrak B_{\kk-1}^\transp \mathfrak B_{\kk-1} + \mathfrak B_{\kk} \mathfrak B_{\kk}^\transp~~,
$$
where $\mathfrak B_\kk \coloneqq D_\kk^{-1/2} B_\kk D_{\kk+1}^{1/2}$. Here $D_\ell = \diag(\vert B_\ell \vert D_{\ell+1} \1)$ for $\ell=\kk-1, \kk$, and $D_{\kk+1}$ is the diagonal matrix of $(\kk+1)$-hull weights. 

From $L_\kk^\mathrm{sym}$ we may define the usual \emph{weighted unnormalized}, and \emph{weighted random-walk normalized $\kk$-Laplacians} $L_\kk^\mathrm{wt}$ and $L_\kk^\mathrm{rw}$, via the formulas:
$$
L_\kk^\mathrm{wt} \coloneqq D_\kk^{1/2} L_\kk^\mathrm{sym} D_\kk^{1/2} \quad\mathrm{ and } \quad L_\kk^\mathrm{rw} \coloneqq D_\kk^{-1} L_\kk^\mathrm{wt}~~.
$$
While in the combinatorial case, $L_\kk$ vanishes for pairs $\sigma\simeq\tau$, each of the weighted Laplacians may be nonzero whenever $\sigma\sim\tau$.

The signed weighted adjacency matrices $S_\kk^\mathrm{sym}$, $S_\kk^\mathrm{wt}, S_\kk^\mathrm{rw}$ are defined analogously to $S_\kk$, as the negative of the off-diagonal parts of their respective Laplacians.
Figure \ref{fig:twotriangleoriented} demonstrates $S_\kk$, $S_\kk^\mathrm{sym}$, $S_\kk^\mathrm{wt}$, and $S_\kk^\mathrm{rw}$ for a simple complex.

\begin{figure}
\hspace{-0.5em}
    \includegraphics[width=.33\textwidth]{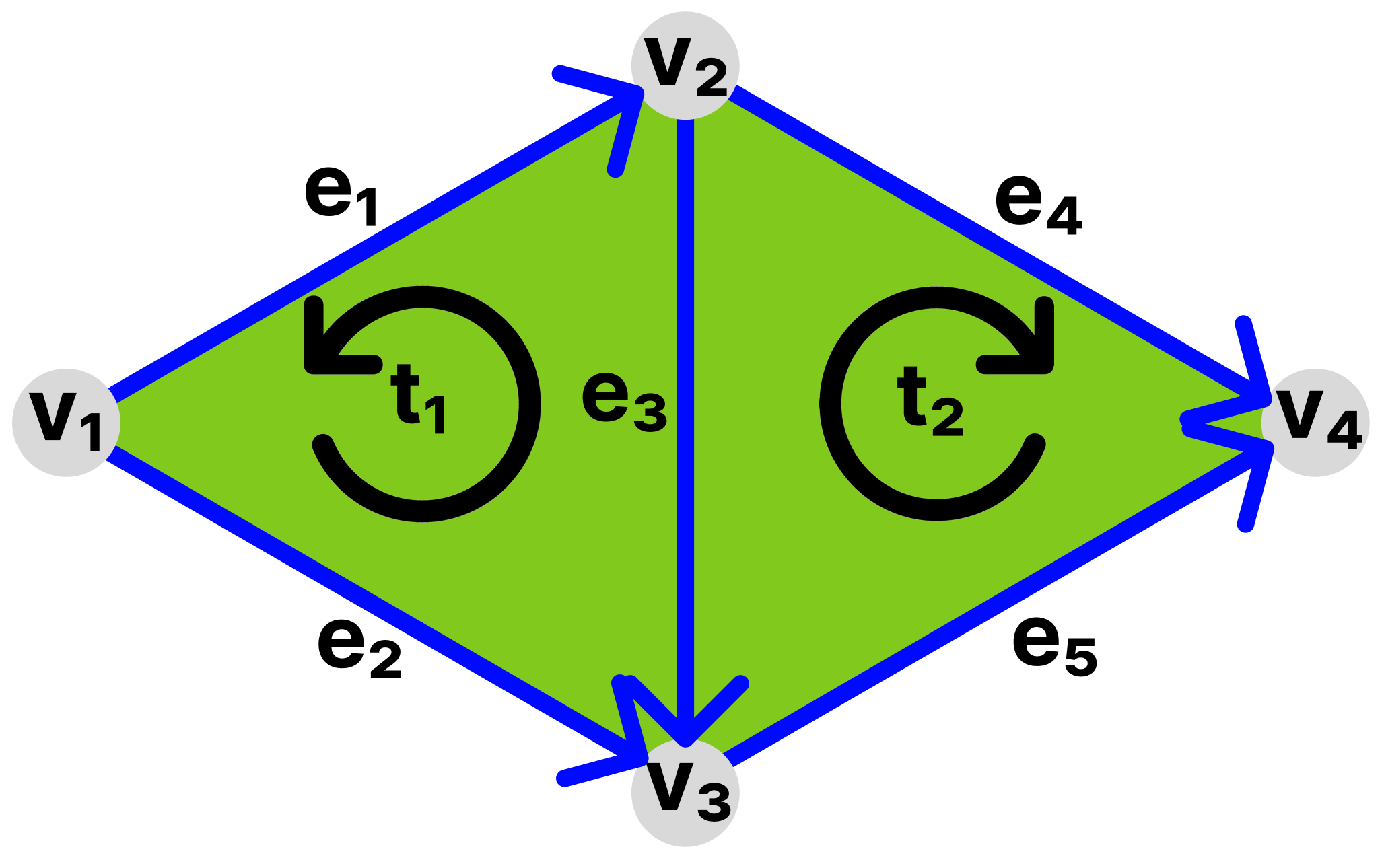}    \hspace{.01\textwidth}
    \raisebox{\height}{$S_1 = \begin{bmatrix}
        0 & 0 & 0 & 1 & 0 \\
        0 & 0 & 0 & 0 & 1 \\
        0 & 0 & 0 & 0 & 0 \\
        1 & 0 & 0 & 0 & 0 \\
        0 & 1 & 0 & 0 & 0
    \end{bmatrix}$}
     \hspace{.01\textwidth}
    \raisebox{\height}{$S_1^\mathrm{sym} = \frac14 \begin{bmatrix}
        0 & 2 & -\sqrt2 & 1 & 0 \\
        2 & 0 & \sqrt2 & 0 & 1 \\
        -\sqrt2 & \sqrt2 & 0 & \sqrt2 & -\sqrt2 \\
        1 & 0 & \sqrt2 & 0 & 2 \\
        0 & 1 & -\sqrt2 & 2 & 0
    \end{bmatrix}$}
    \raisebox{\height}{$S_1^\mathrm{wt} = \frac14 \begin{bmatrix}
        0 & 2 & -2 & 1 & 0 \\
        2 & 0 & 2 & 0 & 1 \\
        -2 & 2 & 0 & 2 & -2 \\
        1 & 0 & 2 & 0 & 2 \\
        0 & 1 & -2 & 2 & 0
    \end{bmatrix}$}
    \hspace{.01\textwidth}
    \raisebox{\height}{$S_1^\mathrm{rw} = \frac14 \begin{bmatrix}
        0 & 2 & -2 & 1 & 0 \\
        2 & 0 & 2 & 0 & 1 \\
        -1 & 1 & 0 & 1 & -1 \\
        1 & 0 & 2 & 0 & 2 \\
        0 & 1 & -2 & 2 & 0
    \end{bmatrix}$}
    \caption{The complex from Figure \ref{fig:twotriangle} on the left, with natural orientation displayed as directed edges and oriented triangles, together with four signed adjacency matrices, the combinatorial $S_1$, the weighted symmetrically normalized $S_1^\mathrm{sym}$, the weighted unnormalized $S_1^\mathrm{wt}$, and the weighted random-walk normalized $S_1^\mathrm{rw}$, all with $D_2 = I$. Notice that the weighted variants may have nonzero entries of various sign even for strongly adjacent simplices, unlike $S_1$.}
    \label{fig:twotriangleoriented}
\end{figure}

\section{Cuts, Fiedler Vectors, and Hierarchical Bipartitions}
\label{sec:Fiedler}

\subsection{Fiedler Vector}
Let $C$ be a simplicial complex, such that $G = (C_0, C_1)$ is a connected graph.
For a given $\kk$, let $\p$ be a vector of orientations over $C_\kk$, with each $[ \p ]_\sigma = p_\sigma \in \{ \pm 1\}$, and let $P = \diag(\p)$.
Let $L_\kk^\mathrm{wt}$, $\tildec L_\kk^\mathrm{wt}$ denote the weighted $\kk$-Laplacian of $C_\kk$ with natural orientations, and with orientations given by $\p$, respectively. 
Let $\lambda_0 \leq \cdots \leq \lambda_{n-1}$ be the eigenvalues of $L_\kk^\mathrm{wt}$ and $\bphi_0, \bphi_1, \ldots, \bphi_{n-1}$ be the corresponding eigenvectors where $n= \vert C_\kk \vert$.
Then, let $(\tildec \lambda_i, \tildec \bphi_i)$ be the eigenpairs for $\tildec L_\kk^\mathrm{wt}$.
Because $\tildec L_\kk^\mathrm{wt} = P L_\kk^\mathrm{wt} P$, $\tildec \lambda_i = \lambda_i$ and $\tildec \bphi_i = P\bphi_i$ for $0\le i < n$.

For $\kk=0$, with the vertices of $G$ in natural orientation, we have that $\lambda_0 = 0$, $\lambda_1 > 0$, $\bphi_0 = \1$ and in particular is non-oscillatory, and that $\bphi_1$ acts as a single global oscillation, appropriate to partition the vertices of $G$ with.
Considering $\tildec L_0^\mathrm{wt}$ for nontrivial $\p \ne \pm \1$, $\tildec \bphi_0$ is oscillatory, and $\tildec \bphi_1$ is no longer appropriate for clustering; this is one reason that oriented $0$-simplices are always considered to be in natural orientation. 

For $\kk >0$ however, it is no longer true that $\bphi_0$ will be non-oscillatory.
Let $\p^*$ be a vector of orientations such that where $[\bphi_0]_\sigma \ne 0$, $[ \p^* ]_\sigma = \sign([\bphi_0]_\sigma)$.
Then the corresponding $\tildec \bphi_0$ \emph{is} non-oscillatory, and acts as a DC component.
This motivates taking $\sign(\bphi_0) \odot \bphi_1$ (element-wise) as the Fiedler vector of $L_\kk^\mathrm{wt}$, with which to partition $C_\kk$.

We will aim to bipartition $\kk$-regions by following a standard strategy in spectral clustering, of minimizing a relaxation of a combinatorial cut function over possible partitions.
Just as a graph cut is typically defined as the volume of edge weight which crosses a partition of the vertices, we can define the \emph{consistency cut} of
$C_\kk$ into subregions $A, B$ as
$$
\Ccut(A, B) \coloneqq \sum_{\substack{\sigma \in A, \tau \in B\\ \sigma \sim \tau} }[S_\kk^\mathrm{wt}]_{\sigma\tau}~~.
$$
Because of the signs introduced by consistency, we consider $\S_\kk^\mathrm{wt}$ as the signed, weighted adjacency matrix for a signed graph over $C_\kk$, and so can utilize the framework of \emph{signed} Laplacians~\cite{kunegis2012signedlaplacian}. Let $[S_\kk^+]_{\sigma\tau} \coloneqq \max(0, [S_\kk^\mathrm{wt}]_{\sigma\tau})$ and $[S_\kk^-]_{\sigma\tau} \coloneqq \max(0, -[S_\kk^\mathrm{wt}]_{\sigma\tau})$, i.e., indicator functions for consistent/inconsistent pairs, respectively. Then, we can define the \emph{consistency volume} $\Cvol^\pm(A) \coloneqq \Ccut^\pm(A, A)$ and the \emph{signed $\kk$-cut}
$$
\Kcut(A, B) \coloneqq 2 \Ccut^+(A, B) + \Cvol^-(A) + \Cvol^-(B)~~.
$$
In the $\kk=0$ case, with all vertices in natural orientation, $S_0^\mathrm{wt}$ is just the usual adjacency matrix, and so $S_0^- = \0$; hence $\Kcut = 2\Ccut$, yielding the traditional cut objective. 
For $\kk > 0$, $\Kcut$ increases with the number of consistent pairs of $\kk$-adjacent simplices across the partition, and with the number of inconsistent pairs within each $\kk$-region.
Equivalently, minimizing $\Kcut$ requires maximizing consistent pairs within each $\kk$-region, and maximizing inconsistent pairs across the partition.

Let $\overline L_\kk$ be the signed Laplacian with signed adjacency $S_\kk^\mathrm{wt}$. 
Let $A$ be a $\kk$-region, $\r_A \coloneqq \1_{A} - \1_{C_\kk\setminus A}$, 
and define $R_A(L) \coloneqq \r_A^\transp L \r_A$.
Then because $\overline L_\kk$ differs from $L_\kk^\mathrm{wt}$ only on the diagonal, $R_A(\overline L_\kk)$ differs from $R_A(L_\kk^\mathrm{wt})$ by a constant independent of $A$.
From \cite{kunegis2012signedlaplacian}, we know that $R_A(\overline L_\kk) \propto \Kcut(A, C_\kk\setminus A)$.
Hence, $\min_{A\subset C_\kk}R_A(L_\kk^\mathrm{wt}) = \min_{A \subset C_\kk}\Kcut(A, C_\kk\setminus A)$, and we obtain $\bphi_0$ as a relaxed solution to $\kk$-cut minimization. 

Now, notice that if the orientations of $C_\kk$ were changed according to some $\p$, this would be equivalent to a different choice of $A$; namely, if $[ \p ]_\sigma= -1$, then $\sigma$ moves to the other side of the partition, either into or out of $A$. As all orientations are available to us, this includes one for which $\tildec \bphi_0$ is non-oscillatory, so that its sign does not partition $C_\kk$. We then instead take $\tildec \bphi_1$ as our relaxed solution, which we may compute via $\sign(\bphi_0) \odot \bphi_1$.

An improved cut objective is the \emph{signed Ratio Cut}, which encourages more balanced partitions:
$$
\SignedRatioCut(A) \coloneqq \left(\frac1{\vert A \vert} + \frac1{\vert C_\kk\setminus A\vert}\right)\Kcut(A, C_\kk\setminus A)~~.
$$
From \cite{kunegis2012signedlaplacian}, we know that with $r_A$ above scaled by a factor of $c_A \coloneqq \sqrt{\vert A \vert / \vert C_\kk \setminus A \vert}$, the analogous result holds, that the eigenvectors of $\overline L_\kk$ yield a relaxed solution to $\min_{A \subset C_\kk}\SignedRatioCut(A)$. However, the new dependence on $A$ means the resulting objective is slightly different for $L_\kk$, so the relaxation is only approximate. 

Finally, the \emph{signed Normalized Cut} balances the partitions by degree rather than simplex count:
$$
\SignedNormalizedCut(A) \coloneqq \left(\frac1{\Cvol(A)} + \frac1{\Cvol( C_\kk\setminus A)}\right)\Kcut(A, C_\kk\setminus A) .
$$
Here, the eigenvectors of $\diag(\overline L_\kk)^{-1} \overline L_\kk$ yield a relaxed solution to $\min_{A\subset C_\kk}~~\SignedNormalizedCut(A)$, and an approximate relaxed solution is given by the eigenvectors of $L_\kk^\mathrm{rw}$. In our numerical experiments, we use the random-walk $\kk$-Laplacian for bipartitioning simplicial complexes, and obtain its eigenvectors from those of $L_\kk^\mathrm{sym}$.




\subsection{Hierarchical Bipartitions}
\label{sec:hp}

The foundation upon which our multiscale transforms on a $\kk$-simplices $C_\kk$ of a given simplicial complex $C$ are constructed is a \emph{hierarchical bipartition tree} (also known as a \emph{binary partition tree}) of $C_\kk$, a set of tree-structured $\kk$-subregions of $C_\kk$ constructed by recursively bipartitioning $C_\kk$. This bipartitioning operation ideally splits each $\kk$-subregion into two smaller $\kk$-subregions that are roughly equal in size while keeping tightly-connected $\kk$-simplices grouped together. 
More specifically, let $C^j_k$ denote the $k^{th}$ $\kk$-subregion on level $j$ of
the binary partition tree of $C_\kk$ and $n^j_k \coloneqq \left\vert C^j_k \right\vert$, where
$j, k \in \ZZ_{\geq 0}$.
Note $C^0_0 = C_\kk$, $n^0_0 = n$, i.e., level $j=0$ represents the root node of
this tree.
Then the two children of $C^j_k$ in the tree, $C^{j+1}_{k'}$ and $C^{j+1}_{k'+1}$,
are obtained through partitioning $C^j_k$ using the Fiedler vector of
$L^\mathrm{rw}_\kk(C^j_k)$. This partitioning is recursively performed until each
subregion corresponding to the leaf contains only a simplex singleton. Note that
$k' = 2k$ if the resulting binary partition tree is a perfect binary tree.
We note that even other (non-spectral) partitioning methods can be used to form the binary partition tree, but in this article, we stick with the spectral clustering using the Fiedler vectors. For more details see on hierarchical partitioning, (specifically for the $\kk = 0$ case), see~\cite[Chap.~3]{IRION-PHD} and \cite{saito2022eghwt}. 
Figure~\ref{fig:hierarchical} demonstrates such a hierarchical bipartition tree for a simple $2$-complex consisting of triangles. 
\begin{figure}
  \centering\includegraphics[width=.9\textwidth]{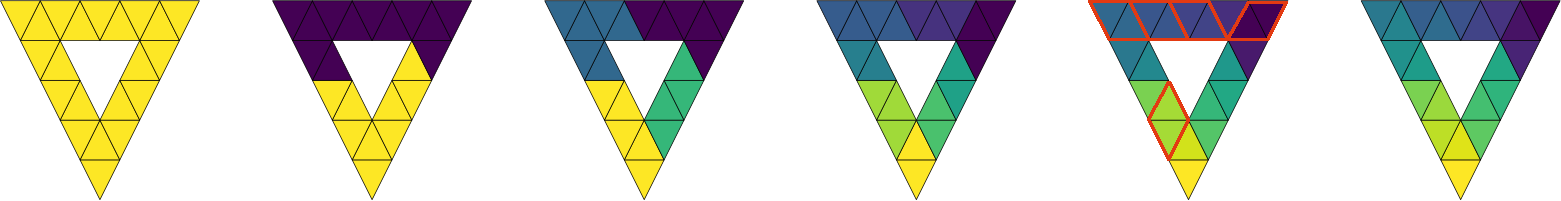}
  \caption{One possible hierarchical bipartitioning of a simple 2-complex, from $j=0$ with no partition on the left, to $j=5$ on the right, where each of the 21 triangles forms their own subregion. Colors indicate distinct subregions. We have highlighted the subregions that contain more than one element for $j=4$.} 
  \label{fig:hierarchical}
\end{figure}

\section{Orthonormal \texorpdfstring{$\kk$}{}-Haar Bases}
\label{sec:Haar}
The classical Haar basis \cite{haar1909theorie} was introduced in 1909 as a piecewise-constant compactly-supported multiscale orthonormal basis (ONB) for square-integrable functions but has since been recognized as a wavelet family and adapted to many domains. In one dimension, the family of Haar wavelets on the interval $[0,1]$ can be generated by the following mother and scaling (or father) functions: \\
\begin{minipage}{.49\linewidth}
\begin{equation*}
        \psi(x)= 
\begin{cases}
    1,      & 0 \leq x < \frac{1}{2}; \\
    -1,      &  \frac{1}{2} \leq x < 1; \\
    0,  & \mathrm{otherwise.}
\end{cases}
\end{equation*}
\end{minipage}
\begin{minipage}{.49\linewidth}
\begin{equation*}
    \phi(x)=
\begin{cases}
    1,      & 0 \leq x < 1; \\
    0,  & \mathrm{otherwise.}
\end{cases}
\end{equation*}
\end{minipage}\\
Unfortunately, these definitions do not generalize to non-homogeneous domains due to the lack of appropriate translation operators and dilation operators \cite{schonsheck2022parallel}. Instead, several methods have been proposed to generate similar bases, and overcomplete dictionaries to apply more abstract domains such as graphs and discretized manifolds \cite{irion2014generalized, shao2019extended, saito2022eghwt}. Here, we describe a method to compute similar, piecewise-constant locally supported bases for $\kk$-simplex-valued functional spaces, which we call the (orthonormal)
\emph{$\kk$-Haar bases}.

Rather than basing our construction on some kind of translation or transportation schemes, we instead employ the hierarchical bipartition, as we discussed in Section~\ref{sec:hp}, to divide the domain, i.e., the $\kk$-simplices $C_\kk$ of a given simplicial complex $C$ into appropriate locally-supported $\kk$-regions. For each $\kk$-region in the bipartition tree, if that region has two children in the tree, then we create a vector that is positive on one child, negative on the other, and zero elsewhere. To avoid sign ambiguity, we dictate that the positive portion is on the region whose region index is smaller among these two. 
See Algorithm~\ref{alg:Haar} for the detail.

Several remarks on this basis are in order. First, since the division is not symmetrically dyadic, we need to compute the scaling factor for each region separately. For each given basis vector $\bxi$ except the scaling vector, we break it into positive and negative parts $\bxi^+$ and $\bxi^-$ and ensure that  $\sum_{i} ([\bxi^+]_i + [\bxi^-]_i) = 0$ and $\|\bxi\|=1$. If the members of $\kk$-region are weighted, then this sum and norm can be computed with respect to those weights. Finally, we note that different hierarchical bipartition schemes may arise from the different weighting of the Hodge Laplacian, which will correspond to bases with different supports.
Figure~\ref{fig:2Haar} demonstrates a $2$-Haar basis based on the partition shown in Figure~\ref{fig:hierarchical}.
\begin{figure}
  \centering\includegraphics[width=.8\textwidth]{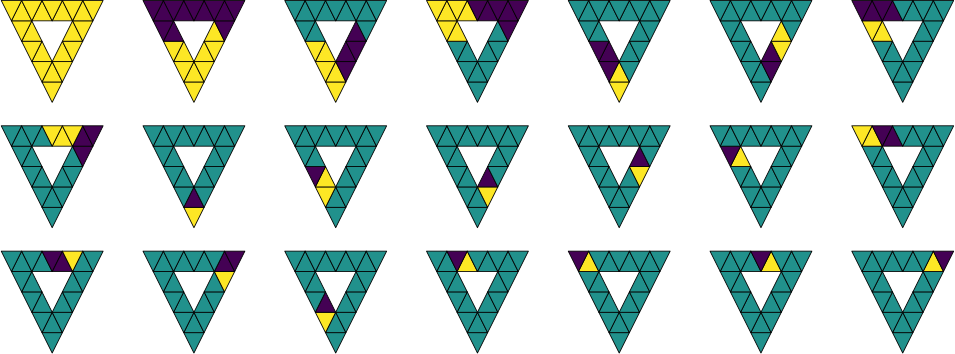}
  \caption{The $2$-Haar basis vectors for the bipartition shown in  Figure~\ref{fig:hierarchical}. The yellow, dark green, violet regions in each vector indicate its positive, zero, and negative components.}
  \label{fig:2Haar}
\end{figure}

\begin{algorithm2e}
\caption{Generating $\kk$-Haar Basis} \label{alg:Haar}
\KwData{ $\left\{ C^j_k\right\}_{j,k}$: A hierarchical bipartition tree of the $\kk$-simplices $C_\kk$ as defined in Section~\ref{sec:hp}, $K^j$ denotes the number of subregions on the level $j$, $n^j_k \coloneqq \left\vert C^j_k \right\vert$ }
\KwResult{An unnormalized $\kk$-Haar Basis $\{\bxi_i \}_{i=0}^{n-1}$ }
 Set $\bxi_0 = \textbf{1}$, $i = 1$\;
\For{$j \coloneqq 1, \dots ,\jmax-1$}{
    \For{$k \coloneqq 0, \dots , K^j-1$}{
        \uIf{$n^j_k = 3$}{
            Set $\bxi_{i} = \textbf{1}_{C^{j+1}_{k'}} - \textbf{1}_{C^{j+1}_{k'+1}}$ \;
            Set $\bxi_{i+1} = \textbf{1}_{C^j \setminus (C^{j+1}_{k'} \cup C^{j+1}_{k'+1})}$\;
            $i = i + 2$\;
        }
        \uElseIf{$n^j_k=1$}{
            Do nothing
        }
        \uElse{\
            Set $\bxi_i = \textbf{1}_{C^{j+1}_{k'}} - \textbf{1}_{C^{j+1}_{{k'}+1}}$ \;
            Set $i = i + 1$ \;
        }
    }
}
\end{algorithm2e}

\section{Overcomplete Dictionaries}
\label{sec:dict}
In this section, we introduce two overcomplete dictionaries for analyzing real-valued functions defined on $\kk$-simplices in a given simplicial complex: the \emph{$\kk$-Hierarchical Graph Laplacian Eigen Transform} ($\kk$-HGLET), based on the \emph{Hierarchical Graph Laplacian Eigen Transform} (HGLET)~\cite{irion2014hierarchical} and the \emph{$\kk$-Generalized Haar-Walsh Transform} ($\kk$-GHWT), based on the \emph{Generalized Haar-Walsh Transform} (GHWT)~\cite{irion2014generalized} for graph signals.

\subsection{\texorpdfstring{$\kk$}{}-Hierarchical Graph Laplacian Eigen Transform (\texorpdfstring{$\kk$}{}-HGLET)}
The first overcomplete transform we describe can be viewed as a generalization of the Hierarchical Block Discrete Cosine Transform (HBDCT). The classical HBDCT is generated by creating a hierarchical bipartition of the signal domain and computing the DCT of the local signal supported on each subdomain. 
We note that a specific version of the HBDCT (i.e., a homogeneous split of an input image into a set of blocks of size $8 \times 8$ pixels) has been used in the JPEG image compression standard~\cite{JPEG}.
This process was generalized to the graph case in \cite{irion2014hierarchical}, i.e., the \emph{Hierarchical Graph Laplacian Eigen Transform} (HGLET), from which we base our algorithm and notation. In turn, our $\kk$-HGLET is a generalization of the HGLET for $\kk$-simplices in a given simplicial complex.
We organize this dictionary by grouping the elements into $\jmax +1$ orthonormal matrices $\left\{\Phi^j\right\}_{j=0}^{\jmax}$ where $\Phi_j \in \RR^{n \times n}$ represents the orthonormal basis formed from the $j$th level of the bipartition. More specifically, let $\{\bphi^j_{k,l} \}$ be the basis vectors in the $\kk$-HGLET where $j$ denotes the level of the partition (with $j=0$ being the root), $k$ indicates the partition within the level, and $l$ indexes the elements within each partition in increasing frequency.

To compute the transform, we first compute the complete set of eigenvectors $\{\bphi^0_{0,l} \}_{l=0:n-1}$ of the Hodge Laplacian of the entire $\kk$-simplices $C_\kk$ of a given simplicial complex $C$ and order them by nondecreasing eigenvalues. We then partition $C_\kk$ into two disjoint $\kk$-regions $C^1_0$ and $C^1_1$ as described in Section~\ref{sec:Fiedler}. We then compute the complete set of eigenvectors of the Hodge Laplacian on $C^1_0$ and $C^1_1$. We again order each set by nondecreasing frequency (i.e., eigenvalue) and label these $\{\bphi^1_{0,l} \}_{l=0:n^1_0-1}$ and  $\{\bphi^1_{1,l} \}_{l=0:n^1_1-1}$ Note that $n^1_0+n^1_1=n^0_0=n$, and that all of the elements in  $\{\bphi^1_{0,l} \}$ are orthogonal to those in $\{\bphi^1_{1,l} \}$ since their supports are disjoint. Then the set $\{\bphi^1_{0,l} \}_{l=0:n^1_0-1} \cup \{\bphi^1_{1,l} \}_{l=0:n^1_1-1}$ form an orthonormal basis for vectors on $C_\kk$ (after extending these vectors beyond their support by zeros). From here, we apply this process recursively, generating an orthonormal basis for each level in the given hierarchical bipartition tree. This process is detailed in Algorithm~\ref{alg:HGLET}.

If the hierarchical bipartition tree terminates at every region containing only a $\kk$-simplex singleton, then the final level will simply be the standard basis of $\RR^n$. Each level of the dictionary contains an ONB whose vectors have the support of roughly half the size of the previous level. There are roughly $(1.5)^n$ possible ONBs formed by selecting different covering sets of regions from the hierarchical bipartition tree; see, e.g., \cite{THIELE-VILLEMOES, saito2022eghwt} for more about the number of possible ONBs in such a hierarchical bipartition tree. Finally, we note that the computational cost of generating the entire dictionary is $\mathcal{O}(n^3)$ and that any valid hierarchical bipartition tree can be used to create a similar dictionary. 
Figure~\ref{fig:HGLET} shows the $2$-HGLET constructed on the $2$-complex and bipartition shown in Figure ~\ref{fig:hierarchical} .
\begin{figure}
  \centering
\includegraphics[width=\textwidth]{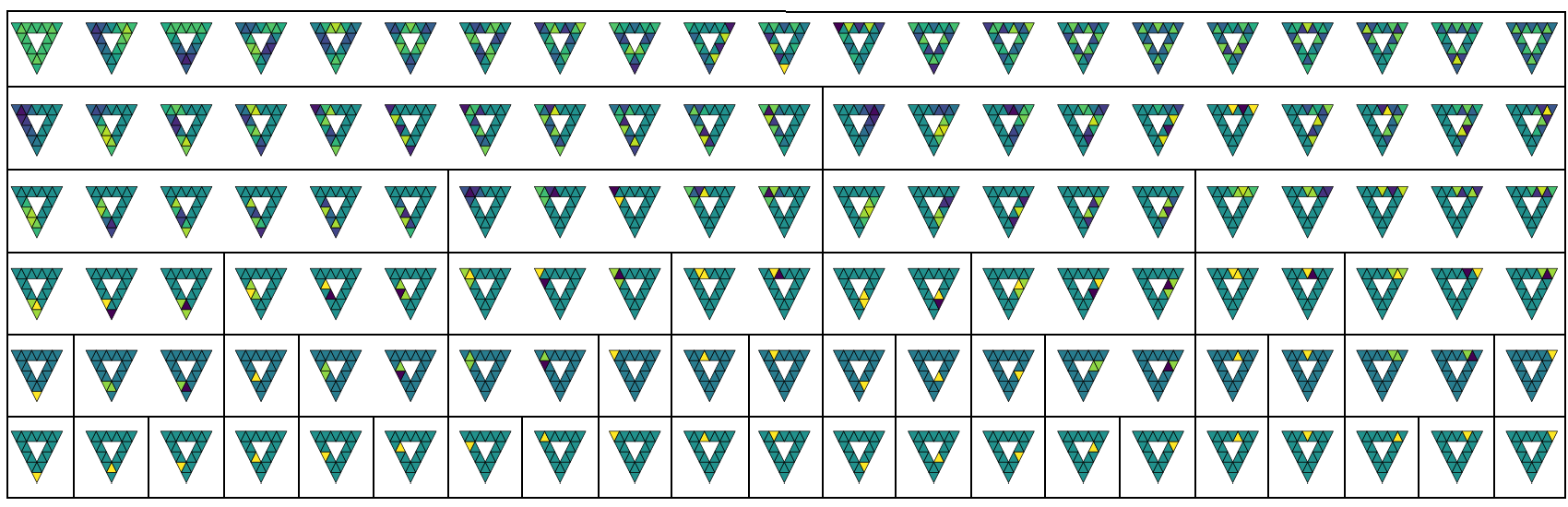}
\caption{$2$-HGLET dictionary a $2$-complex. Each row represents a different level $j$ and the subregions within each level are shown with the boxes. Here, the color scale is consistent across each row to better visualize the smoothness of the elements. Note that we can generate an ONB by selecting any subset of boxes so that the union of those boxes contains exactly one element from each column.}
  \label{fig:HGLET}
\end{figure}
\begin{algorithm2e}
\caption{Generating $\kk$-HGLET Dictionary} \label{alg:HGLET}
\KwData{ $\left\{ C^j_k\right\}_{j,k}$: A hierarchical bipartition tree of the $\kk$-simplices $C_\kk$ as defined in Section~\ref{sec:hp}, $K^j$ denotes the number of subregions on the level $j$, $n^j_k \coloneqq \left\vert C^j_k \right\vert$,  $L_\kk:$  $\kk$-Hodge Laplacian of $C_\kk$ }
\KwResult{A multiscale overcomplete dictionary $\left\{\Phi^j\right\}_{j=0}^{\jmax}$ }
 \For{$j \coloneqq 0, \dots, \jmax$}{
  \eIf{$j = 0$}{
   Compute $\Phi^{0}$ as solution to eigensystem $L_{\kk} \bphi = \lambda \bphi$ \; 
   }{
   \For{$k \coloneqq 0, \dots , K^j-1$}{
   Compute the local $\kk$-Hodge Laplacian $L_\kk(C^j_k)$ \;
   Compute $\hat{\Phi}^{j}_{k}$ as solution to eigensystem $L_{\kk}(C^j_k) \,\hat{\!\bphi} = \lambda \,\hat{\!\bphi}$\;
   Set $\left[\Phi^{j}_{k}\right]_{C^j_k} = \hat{\Phi}^{j}_{k}$ and $\left[\Phi^{j}_{k}\right]_{(C^j_k)^\perp} = \textbf{0}$
  }
  Set $\Phi^{j} = \left( \Phi^{j}_{0}, \dots , \Phi^{j}_{n^j_k-1} \right)$ \;
 }
}
\end{algorithm2e}


\subsection{\texorpdfstring{$\kk$}{}-Generalized Haar-Walsh Transform (\texorpdfstring{$\kk$}{}-GHWT)}
The second transform we present here is based on the \emph{Generalized Haar-Walsh Transform} (GHWT)~\cite{irion2014generalized}, which can itself be viewed as a generalization of the Wash-Hadamard transform. This basis is formed by first generating a hierarchical bipartition tree of $C_\kk$. We then work in a bottom-up manner, beginning with the finest level in which each region only contains a single element. We call these functions scaling vectors and label them $\{ \bpsi^{\jmax}_{k,0} \}_{k=0:n-1}$. For the next level, we first assign a constant scaling vector for support on each region. Then, for each region that contains two children in the bipartition tree, we form a Haar-like basis element by subtracting the scaling function associated with the child element with a higher index from that child element with a lower index. This procedure will form an ONB $\{ \bpsi^{\jmax-1}_{k,l} \}_{k=0:K^{\jmax-1}-1, l=0:l(k)-1}$ (where $K^{\jmax-1}$ is the number of $\kk$-subregions at level $\jmax-1$ and $l(k) = 1$ or $2$ depending on the subregion $k$) whose vectors have support of at most 2. For the next level, we begin by computing the scaling and Haar-like vectors as before. Next, for any region that contains three or more elements, we also compute Walsh-like vectors by adding and subtracting the Haar-like vectors in the children's regions. From here, we form the rest of the dictionary recursively. A full description of this algorithm (for the $\kk=0$ case) is given in \cite{irion2014hierarchical} and we present a generalized version in Algorithm~\ref{alg:GHWT}. For ease of notation, we present the algorithm for an unnormalized dictionary. We would prefer a normalized dictionary in many applications, but the choice of norm can be problem dependent and can be done by simply looping thought the elements once the unnormalized dictionary has been created.

Figure~\ref{fig:C2F} displays the $2$-GHWT dictionary on
the same $2$-complex used in Figures~\ref{fig:2Haar} and \ref{fig:HGLET}.
\begin{figure}
  \includegraphics[width=\textwidth]{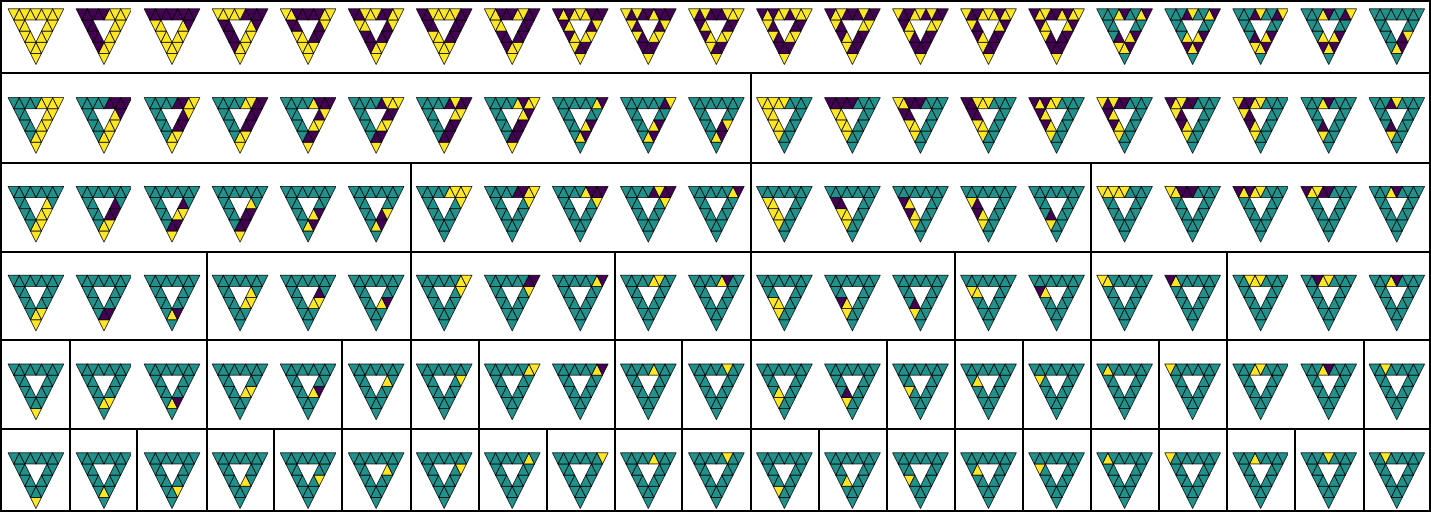}
  \caption{Coarse-to-Fine (C2F) 2-GHWT dictionary. The yellow, dark green, and violet regions in each vector indicate its positive, zero, and negative components, respectively. Each row represents a different level $j$ and the subregions within each level are shown with the boxes. Similarly to the 2-HGLET dictionary in Figure~\ref{fig:HGLET}, we can generate an ONB by selecting any subset of boxes so that the union of those boxes contains exactly one element from each column.}
  \label{fig:C2F}
\end{figure}
We make several observations about this dictionary. First, like the $\kk$-HGLET, each level of the dictionary forms an ONB, and at each level, basis vectors have the support of roughly half the size of the previous level. These basis vectors also have the same support as the $\kk$-HGLET basis vectors (that is, $\supp(\bphi^j_{k,l}) = \supp(\bpsi^j_{k,l})$ for all $j, k, l$). However, the computational cost of computing the $\kk$-GHWT is only $\mathcal{O}(n \log n)$ compared to the $\mathcal{O}(n^3)$ of the $\kk$-HGLET. 

Finally, we note that at the coarsest level $(j=0)$ the $\kk$-GHWT dictionary contains globally-supported piecewise-constant basis vectors, which are ordered by increasing oscillation (or ``sequency''). This forms an ONB analogous to the classical Walsh Basis. This allows us to define an associated Walsh transform and conduct Walsh analysis on signals defined on simplicial complexes. Although not the primary focus of this article, we conduct some numerical experiments using the Walsh bases explicitly in Section~\ref{sec:numexp}. 

\begin{algorithm2e}
\caption{Generating $\kk$-GHWT Dictionary} \label{alg:GHWT}
\KwData{  $\left\{ C^j_k\right\}_{j,k}$: A hierarchical bipartition tree of the $\kk$-simplices $C_\kk$ as defined in Section~\ref{sec:hp}, $K^j$ denotes the number of subregions on the level $j$, $n^j_k \coloneqq \left\vert C^j_k \right\vert$ \\ }
\KwResult{An (unnormalized) $\kk$-GHWT dictionary $\{\bpsi^j_{k,l}\}$ }
\For{$j \coloneqq \jmax, \dots ,1$}{
    \eIf{$j = \jmax$}{
        \For{$k=0, \dots, n-1$}{
            Set $\bpsi^{\jmax}_{k,0} = \textbf{1}_{C^{\jmax}_k}$ \;
        }
    }{
        \For{$k \coloneqq 0,\dots,K^{j-1}-1$}{
            Set $\bpsi^{j-1}_{k,0} = \textbf{1}_{C^j_k}$ \;
            \If{$n^{j-1}_k>1$}{
                Set $\bpsi^{j-1}_{k,1} = \bpsi^j_{k',0} - \bpsi^j_{k'+1, 0}$ \; 
            }
            \If{$n^{j-1}_k>2$}{
                \For{$l = 1,\dots,2^{\jmax-j}-1$}{
                    \uIf{both subregions $k'$ and $k'+1$ have a basis vector with tag $l$}{
                        Set $\bpsi^{j-1}_{k,2l} = \bpsi^j_{k',l} + \bpsi^j_{k'+1, l}$ \;
                        Set $\bpsi^{j-1}_{k,2l+1} = \bpsi^j_{k',l} - \bpsi^j_{k'+1, l}$ \;
                    }
                    \uElseIf{only one subregion has a basis vector with tag $l$}{
                        Set $\bpsi^{j-1}_{k,2l} = \bpsi^j_{k',l}$ \;
                    }
                    \uElseIf{neither subregion has a basis vector with tag $l$}{
                        Do nothing
                    }
                    
                }
            }
        }
    }
}
\end{algorithm2e}

\subsection{Organizing the Dictionaries}
For many downstream applications, it is important to organize the order of these bases. In general, the $\kk$-HGLET dictionary is naturally ordered in a \emph{Coarse-to-Fine} (C2F) fashion. In each region, the basis vectors are ordered by frequency (i.e., eigenvalue). Similarly, the GHWT dictionary is also naturally ordered in a C2F fashion, with increasing ``sequency'' within each subgraph. Another useful way to order the GHWT is in a \emph{Fine-to-Coarse} (F2C) ordering, which approximates ``sequency'' domain partitioning. See, e.g., Figure~\ref{fig:F2C}, which shows the F2C $2$-GHWT dictionary on the same 2-complex used in Figure~\ref{fig:C2F}.
We also note that the F2C ordering is not possible for the $\kk$-HGLET dictionary because some parent subregions and the direct sum of their children subregions are not equivalent; see, e.g., \cite[Eq.~(5.6)]{IRION-PHD} for the details.
Other relabeling schemes, such as those proposed in \cite{shao2019extended, saito2022eghwt} may also be useful but are outside the scope of this article and will be explored further in our future work.
\begin{figure}
  \centering
  \includegraphics[width=\textwidth]{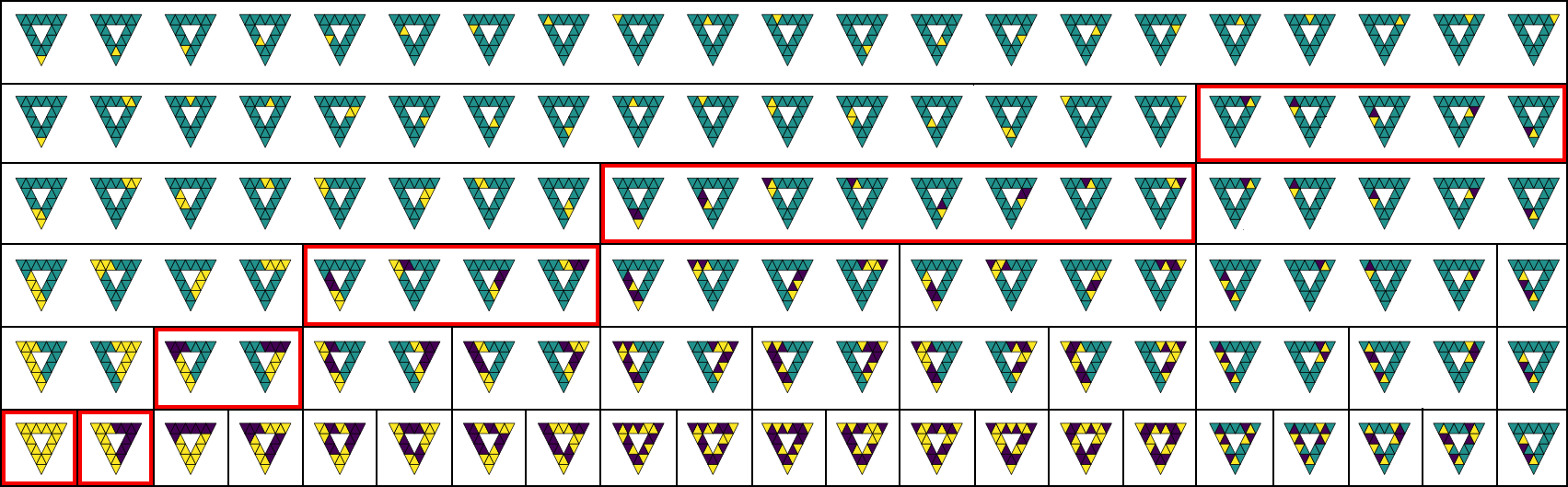}
  \caption{Fine-to-Coarse  (F2C) 2-GHWT dictionary. Note that this dictionary is not generated by simply reversing the row indices of the C2F dictionary, but by also arranging each level (row) by ``sequency.'' That is, each row is first sorted by tag $l$ (shown by the black boxes) then each tag is sorted by region index $k$. The vectors enclosed by the red boxes form the Haar basis for this 2-complex while the vectors in the bottom row form the Walsh basis.}
  \label{fig:F2C}
\end{figure}

\subsection{Basis and Frame Selection}
Once we have established these arrangements of basis vectors, we can efficiently apply the best-basis algorithm~\cite{coifman1992entropy} to select an ONB that is optimal for a task at hand for a given input signal or a class of input signals; see also our previous work of applying the best-basis algorithm in the graph setting~\cite{irion2014hierarchical,irion2014generalized,IRION-SAITO-SPIE,IRION-SAITO-TSIPN,shao2019extended,CLONINGER-LI-SAITO,saito2022eghwt}. Given some cost function $\mathcal{F}$ and signal $\x$, we traverse the bipartition tree and select the basis that minimizes $\mathcal{F}$ restricted to each region. For the C2F dictionary, we initialize the best basis as the finest $(j=\jmax)$ level of the GHWT dictionary. We then proceed upward one level at a time and compute the cost of each subregion at that level and compare it to the cost of the union of its children subregions. If the latter cost is lower, the basis is updated; if not, the children subregions (and their basis vectors) are propagated to the current level. This algorithm yields the C2F best basis. The F2C best basis is performed similarly, i.e., we begin with the globally-supported basis $(j=0)$ at the bottom of the rearranged tree and proceed in the same bottom-up direction. 
As for the HGLET dictionary, it has only a C2F basis as we discussed earlier.

In some contexts, it is not necessary to generate a complete ONB, but rather some sparse set of vectors in the dictionary (also known as atoms) that most accurately approximate a given signal or class of signals. In this case, we can directly apply the orthogonal matching pursuit of \cite{cai2011orthogonal} to find the best $m$-dimensional orthogonal subframe ($m \leq n$) selected from the dictionary. Additionally, for some downstream tasks, such as sparse approximation or sparse feature selection, generating orthogonal sets of atoms is not critical. In these cases, we can employ a greedy algorithm to generate efficient approximation. This algorithm simply selects the atoms in the dictionary with the largest coefficient, removes it, then computes the transform of the residual and proceeds so forth. This algorithm is quite expensive since it need to recompute the coefficients after each selection. Therefore, it is only suited for tasks when the number of elements are small, or we only need to compute a few features. These algorithms are studied extensively in the subsequent section.

\subsection{Approximation Theory}\label{sec:Approx}

Signal approximation has been one of the most important applications for classical and graph wavelet bases. Theoretical justification of the efficacy often relies on developing approximation bounds and decay rates for the wavelet coefficients for various classes of signals under various norms (or seminorms). A number of these results have been specifically developed for graphs equipped with a hierarchical tree \cite{gavish2010multiscale, coifman2011harmonic,sharon2015class, irion2016efficient}. In general, these results are based on first defining a distance function between vertices of a graph, then defining a H\"{o}lder seminorm based on this distance function, and then finally computing the decay rates. By defining a distance between any two $\kk$-simplices as the number of elements in the smallest partition in the tree that constrains both elements, we can generalize these results. Formally, for singleton $\kk$-elements $\sigma$ and $\tau$ of $C_\kk$, and signal $\f$, we define a distance function and then the associated H\"{o}lder seminorm as: 
\begin{equation*}
    d(\sigma, \tau) \coloneqq \min \left\{n^j_k \, \Big\vert \,  \sigma, \tau \in C^j_k \right\}, \quad C_H(\f) \coloneqq \sup_{\sigma \neq \tau} \frac{ \left\vert [\f]_\sigma-[\f]_\tau  \right\vert }{d(\sigma, \tau)^\alpha}
\end{equation*}
where $\alpha$ is a constant in $(0,1]$. With these definitions, the dictionary coefficient decay and approximation results of \cite{sharon2015class, irion2016efficient} for the GHWT and HGLET can be applied to the $\kk$-GHWT and $\kk$-HGLET bases. For sake of space, we only state these theorems and do not reproduce the proofs since the adaptation is trivial after substituting the above definitions.

\begin{theorem}
For a simplicial complex C equipped with a hierarchical bipartition tree, suppose that a signal $\f$ is H\"{o}lder continuous with exponent $\alpha$ and constant $C_H(\f)$. Then the coefficients with $l \geq 1$ for the $\kk$-HGLET ($c^j_{k,l}$) and $\kk$-GHWT ($d^j_{k,l}$) satisfy:
\begin{equation*}
    \left\vert c^j_{k,l} \right\vert \leq C_H(\f) \left(n^j_k\right)^{\alpha+\frac{1}{2}}, \quad  \left\vert d^j_{k,l}  \right\vert \leq C_H(\f) \left(n^j_k\right)^{\alpha+\frac{1}{2}}
\end{equation*} 
\end{theorem}
\begin{proof}
See Theorem~3.1 of \cite{irion2016efficient}.
\end{proof}

\begin{theorem} For a fixed orthonormal basis $\{\bphi_l \}^{n-1}_{l=0}$ and a parameter $0 < \rho < 2$,
\begin{equation*}
    \|\f - P_m \f \|_2 \leq \frac{\vert \f \vert_\rho}{m^\beta}
\end{equation*} 
where $P_m$ in the best nonlinear $m$-term approximation in the basis $\{\bphi_l \}^{n-1}_{l=0}$, $\beta=\frac{1}{\rho} - \frac{1}{2}$ and $\vert \f \vert_\rho$ is defined as
$\vert \f \vert_\rho \coloneqq \left(\sum_{l=0}^{n-1} \vert \langle \f, \bphi_l \rangle \vert^\rho \right)^{\frac{1}{\rho}}$
\end{theorem}
\begin{proof}
See Theorem~3.2 of \cite{irion2016efficient} and Theorem~6.3 of \cite{sharon2015class}.
\end{proof}

\section{Numerical Experiments}
\label{sec:numexp}
We demonstrate the efficacy of our proposed partitioning techniques and basis constructions by conducting a series of experiments. In Section~\ref{subsec:approx} we show how our multiscale bases and overcomplete dictionaries can be used to sparsely approximate signals defined on $\kk$-simplices. In Section~\ref{subsec:classification} we show how these representations can be used in supervised classification and unsupervised clustering problems. 


\subsection{Approximation and Signal Compression}
\label{subsec:approx}

We begin with an illustrative example by creating some synthetic data for 1- and 2-simplices by triangulating a digital image. We start with a $512 \times 512$ ``peppers'' image and map it to a Cartesian grid on the unit square $[0,1]^2$. We then randomly sample $1024$ points within this square (not necessarily on a grid). We then create a triangular mesh from these points using Delaunay triangulation. Next, we interpolate the image from the Cartesian grid to the sampled vertices by computing the barycentric coordinate of each vertex from the square inside the Cartesian grid. Finally, we interpolate the signal to the edges and triangles of the triangulation by averaging the values of the vertices that they contain. The result, for our random seed, is a signal defined on the 3050 edges of the triangulation and another on the 2067 triangles. We now consider the sparse representation of these signals. Figure~\ref{fig:pepper_pics} shows the \emph{nonlinear} approximation (i.e., using the largest expansion coefficients in magnitude) of the triangle-based signals in the Hodge Laplacian eigenbasis (Fourier), the orthonormal Haar basis, orthonormal Walsh basis as well as the approximation prescribed by applying the best-basis and greedy algorithms to the HGLET and GHWT dictionaries. 
Figure~\ref{fig:pepper_approx} shows the approximation error vs the number of terms used for both the edge-based and triangle-based functions. 

A number of observations are in order. First, the multiscale dictionary-based methods consistently outperformed the generic orthonormal bases. The greedy approximation algorithm achieved the best approximation results, but it is also more costly to compute than any of the other methods, and the set of atoms used in the approximation may not be orthogonal. This may be detrimental to downstream tasks. Overall the GHWT-based method performed best, with the F2C best basis performing much better than the C2F best basis, which suggests that the fine-scale features of this signal are the most important. Similarly, the Walsh basis achieved much better results than the Haar basis, again emphasizing the necessity of capturing details at the fine scale.

\begin{figure}
  \centering\includegraphics[width=.97\textwidth]{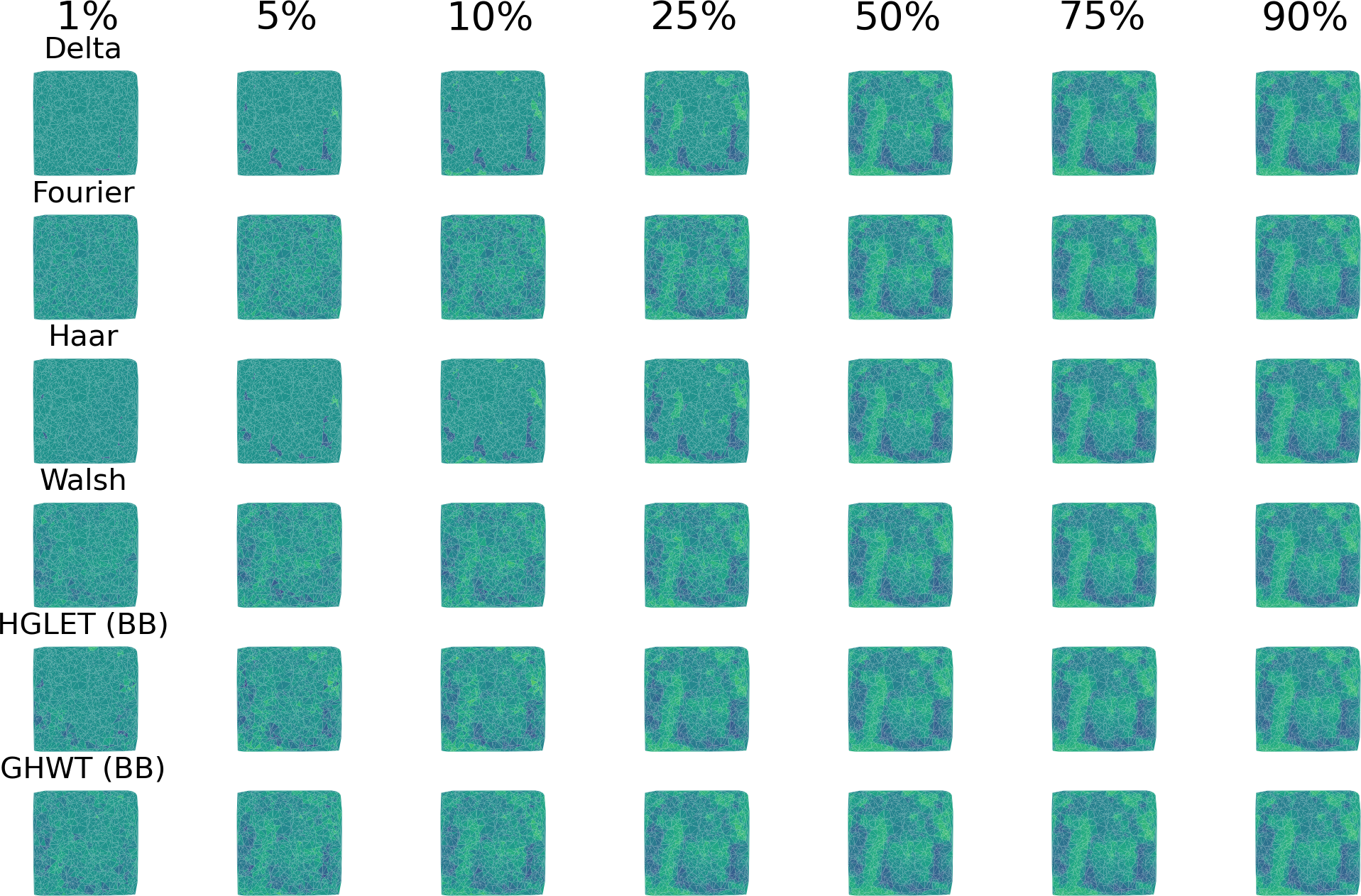}
\caption{ Nonlinear approximation of the peppers image for $\kk=2$}
\label{fig:pepper_pics}
\end{figure}

\begin{figure}
  \centering\includegraphics[width=.97\textwidth]{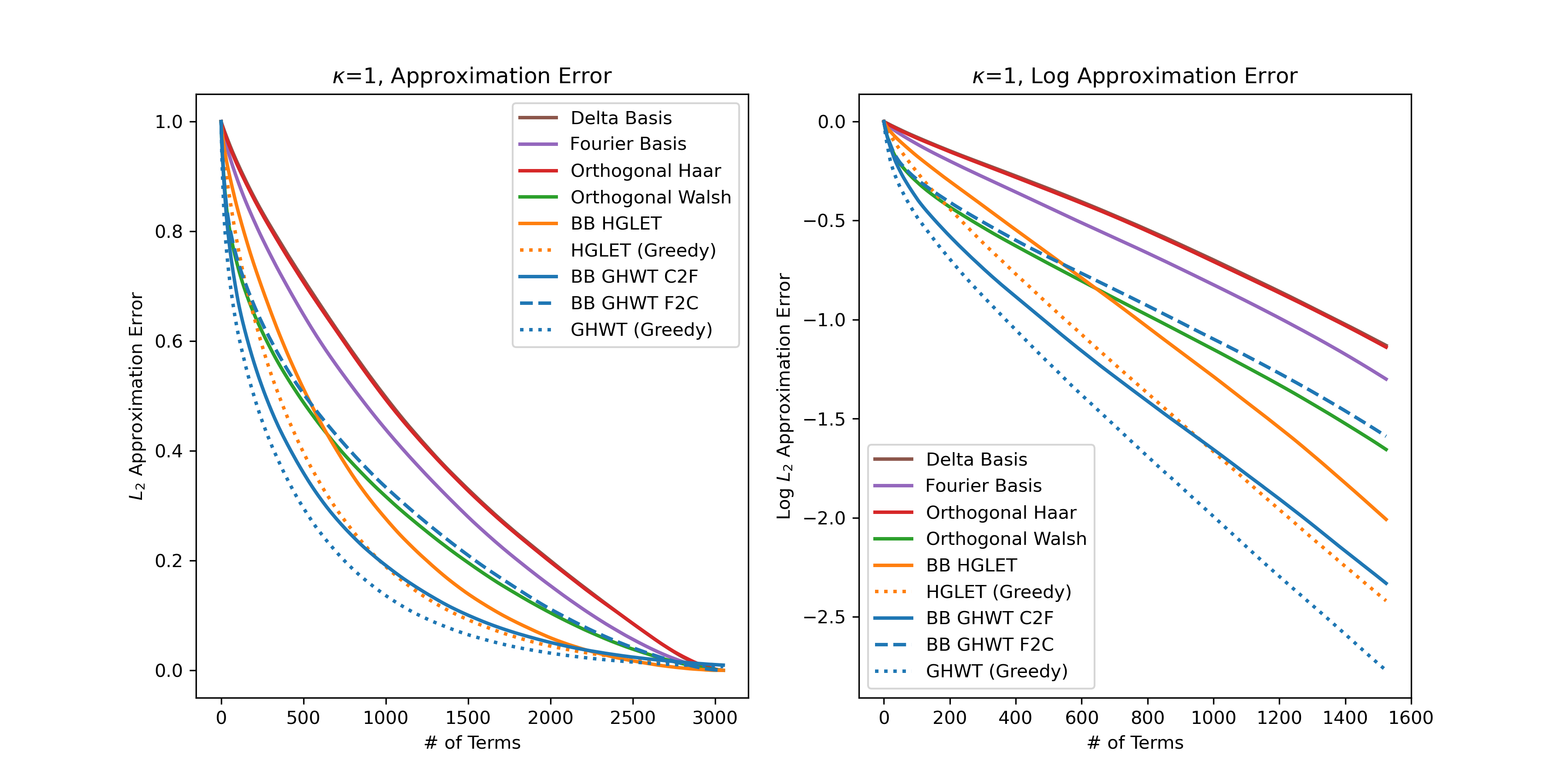}
  \centering\includegraphics[width=.97\textwidth]{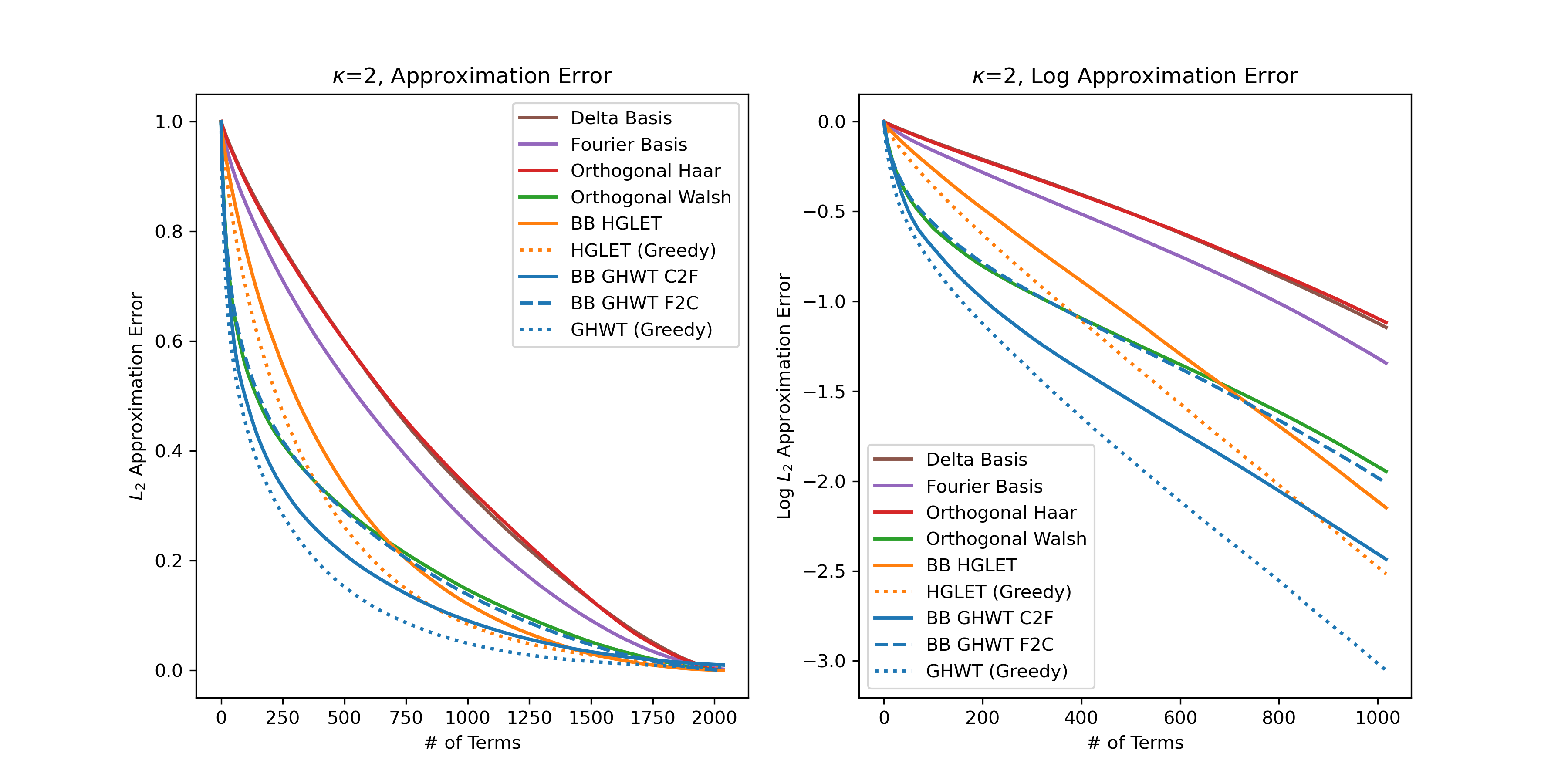}
    \caption{Nonlinear approximation errors of the peppers image, Left: $L_2$ error, Right: $\log(L_2$ error) for up to 50\% of the terms retained. Top $\kk=1$, Bottom: $\kk=2$.}
    \label{fig:pepper_approx}
\end{figure}

Next, we apply our approach to real-world data containing higher-order simplices with $\kk=0, 1, \ldots, 5$. The citation complex \cite{patania2017shape, ebli2020simplicial} is a simplicial complex derived from the co-author/citation complex (CC) \cite{vsubelj2013model}, which models the interactions between multiple authors of scientific papers. A paper with $\kk$ authors is represented by a $(\kk-1)$-simplex. 
We first build a graph whose vertices represent the authors in this CC database. Then, the vertices are connected by edges that represent co-authored papers. Note that if two authors co-authored multiple papers, these two vertices are connected by a single edge. Next, we assign each edge the sum of the citation numbers of all the co-authored papers by the authors, forming this edge as its weight (or value). Finally, we assign each higher-order simplex the sum of the values of its lower-order simplices as its value. 
See \cite{ebli2020simplicial} for a more thorough description of the construction of this complex. Table~\ref{table:CC_stat} reports some basic information about the number of simplices of different dimensions in this citation complex. Figure~\ref{fig:CC_approx1} shows the nonlinear approximation of this signal (i.e., a vector of citation numbers) for $\kk=0, 1, \ldots, 5$ with the Delta, Fourier, Haar, HGLET, and GHWT bases. Figure~\ref{fig:CC_appro2x} shows the log error. The HGLET and GHWT bases were selected by the best-basis algorithm. 

In these experiments, we observe that the best bases (GHWT and HGLET) outperformed the canonical bases, with the GHWT being the most efficient basis for each $\kk$. Additionally, for $\kk > 0$, the orthonormal Haar basis performed best in the semi-sparse regime ($1$ and $10$\% of terms retrained). This suggests that the signals on each dimension of the citation complex are similar in that they are all close to being piecewise constant. However, when more terms are considered, the HGLET best basis achieved a lower approximation error than the orthonormal Haar basis achieved.

\begin{table}
\centering
\begin{tabular}{ccccccc}
\hline
$\kk$       & 0    & 1    & 2     & 3     & 4     & 5     \\ \hline
\# of elements & 352 & 1474 & 3285 & 5019 & 5559 & 4547 \\ \hline
\end{tabular}
\caption{The number of element in the $\kk$-simplices in the citation complex for $\kk=0, 1, \ldots, 5$}
\label{table:CC_stat}
\end{table}

\begin{figure}
    \centering
    \begin{subfigure}[b]{0.32\textwidth}
         \centering
         \includegraphics[width=\textwidth]{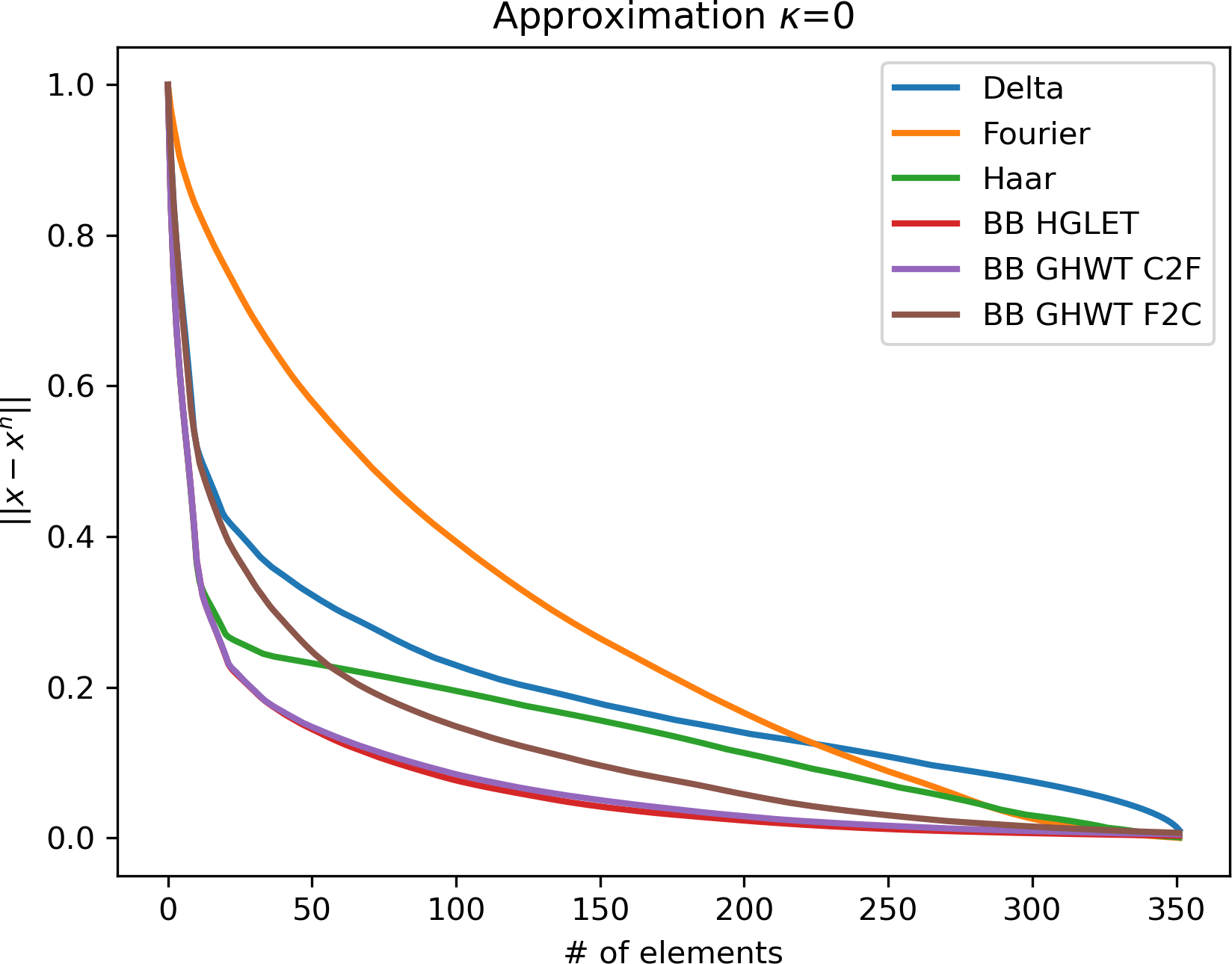}
     \end{subfigure}
     \hfill
     \begin{subfigure}[b]{0.32\textwidth}
         \centering
         \includegraphics[width=\textwidth]{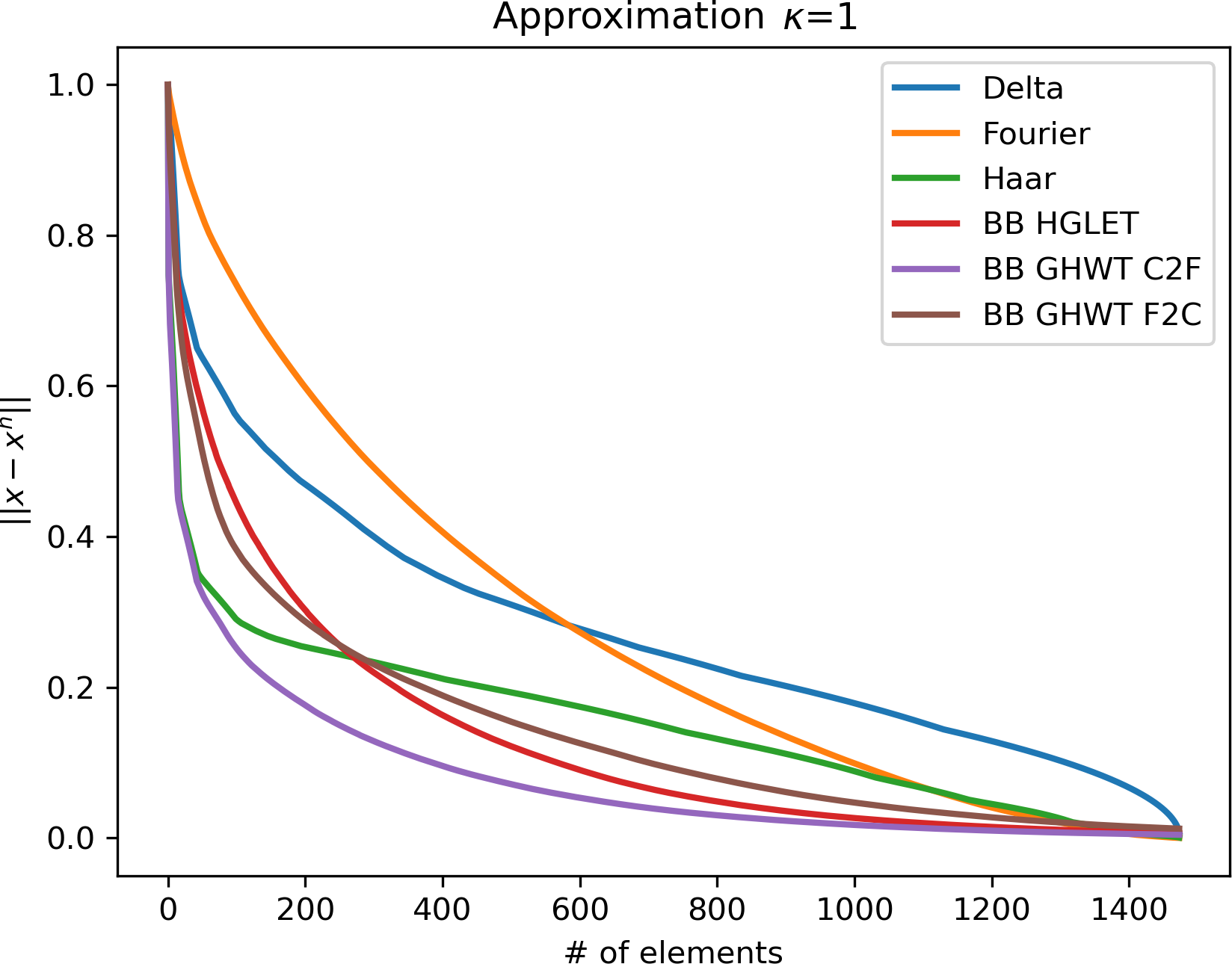}
     \end{subfigure}
     \hfill
     \begin{subfigure}[b]{0.32\textwidth}
         \centering
         \includegraphics[width=\textwidth]{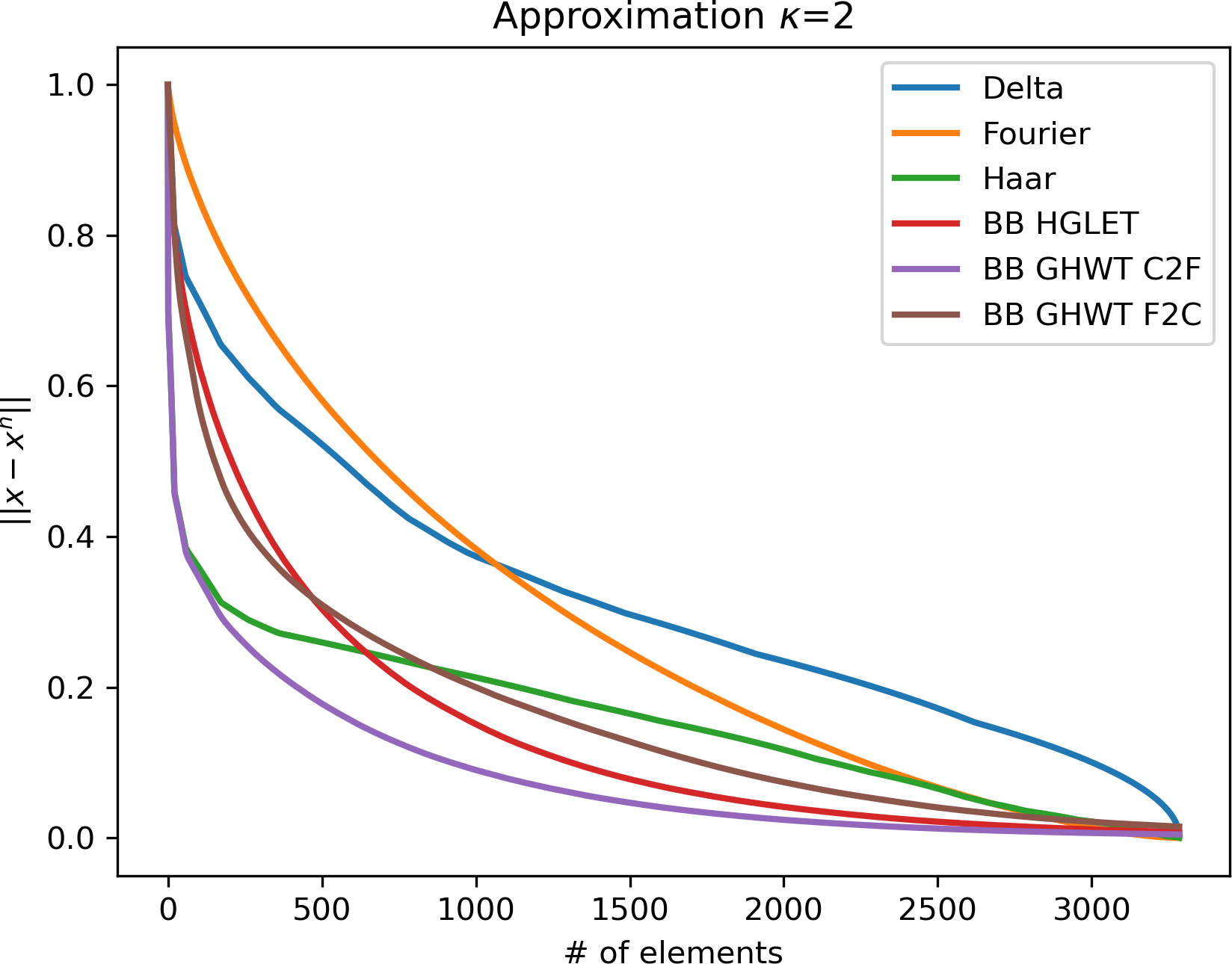}
     \end{subfigure}
     \hfill
     \begin{subfigure}[b]{0.32\textwidth}
         \centering
         \includegraphics[width=\textwidth]{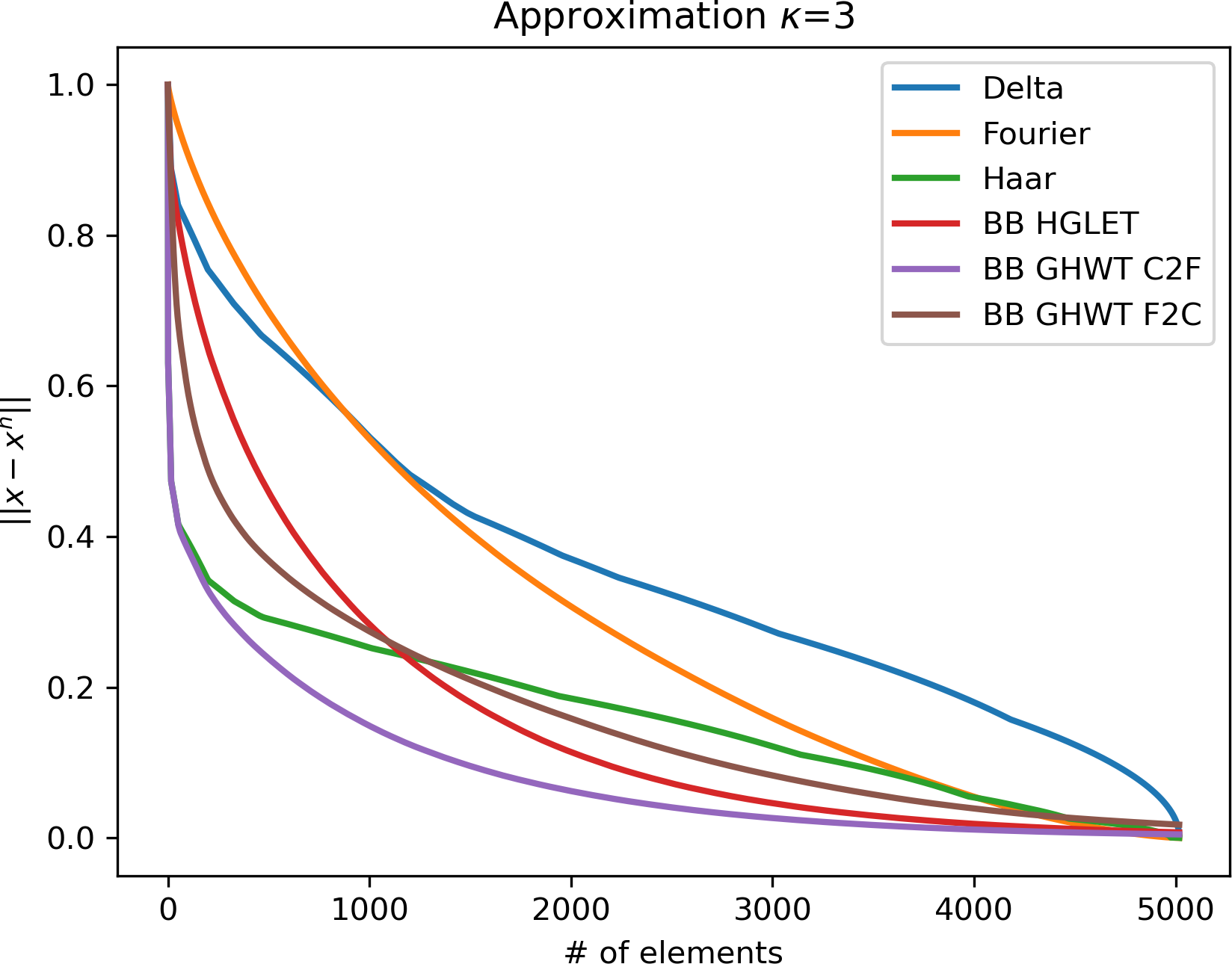}
     \end{subfigure}
     \hfill
     \begin{subfigure}[b]{0.32\textwidth}
         \centering
         \includegraphics[width=\textwidth]{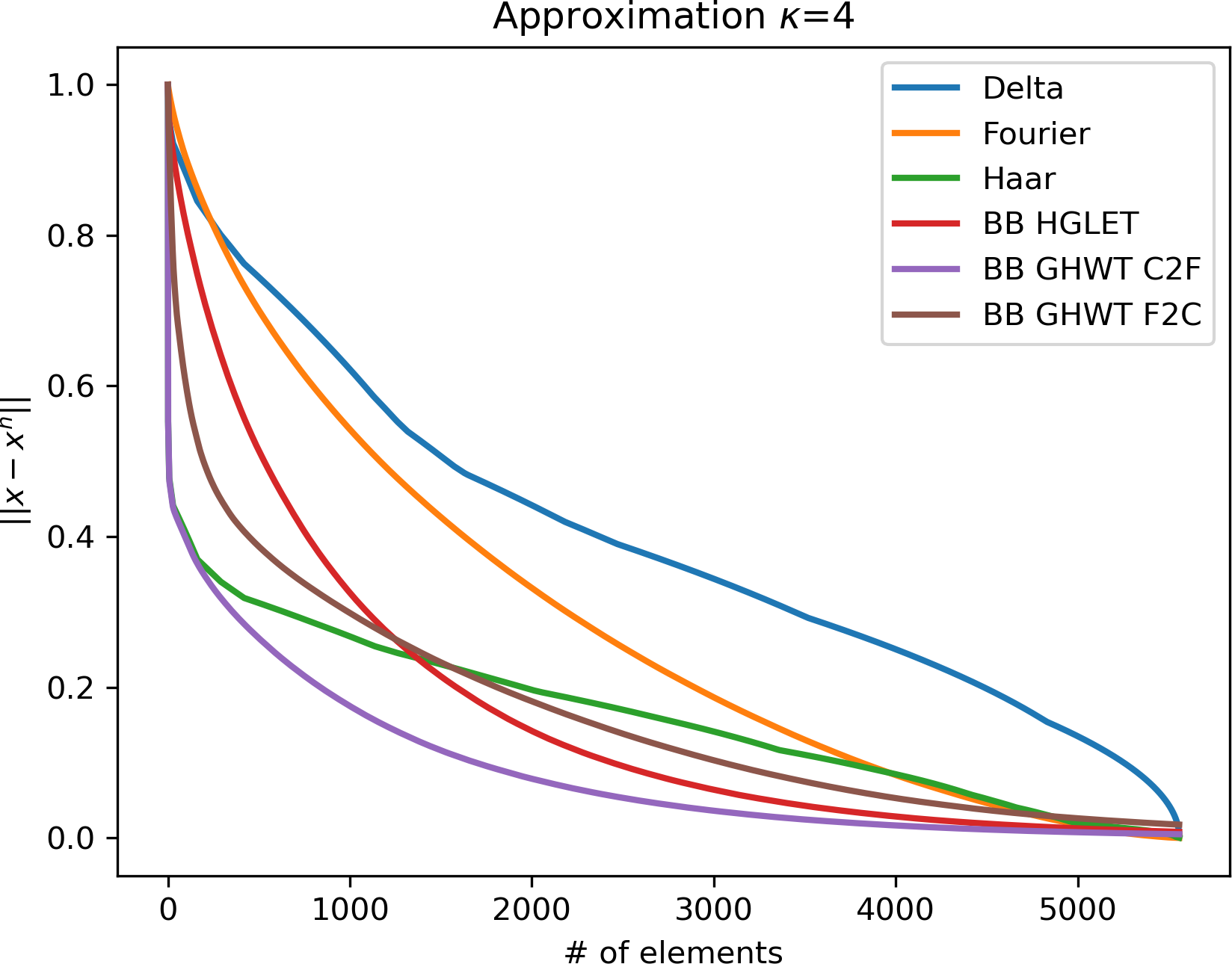}
     \end{subfigure}
     \hfill
     \begin{subfigure}[b]{0.32\textwidth}
         \centering
         \includegraphics[width=\textwidth]{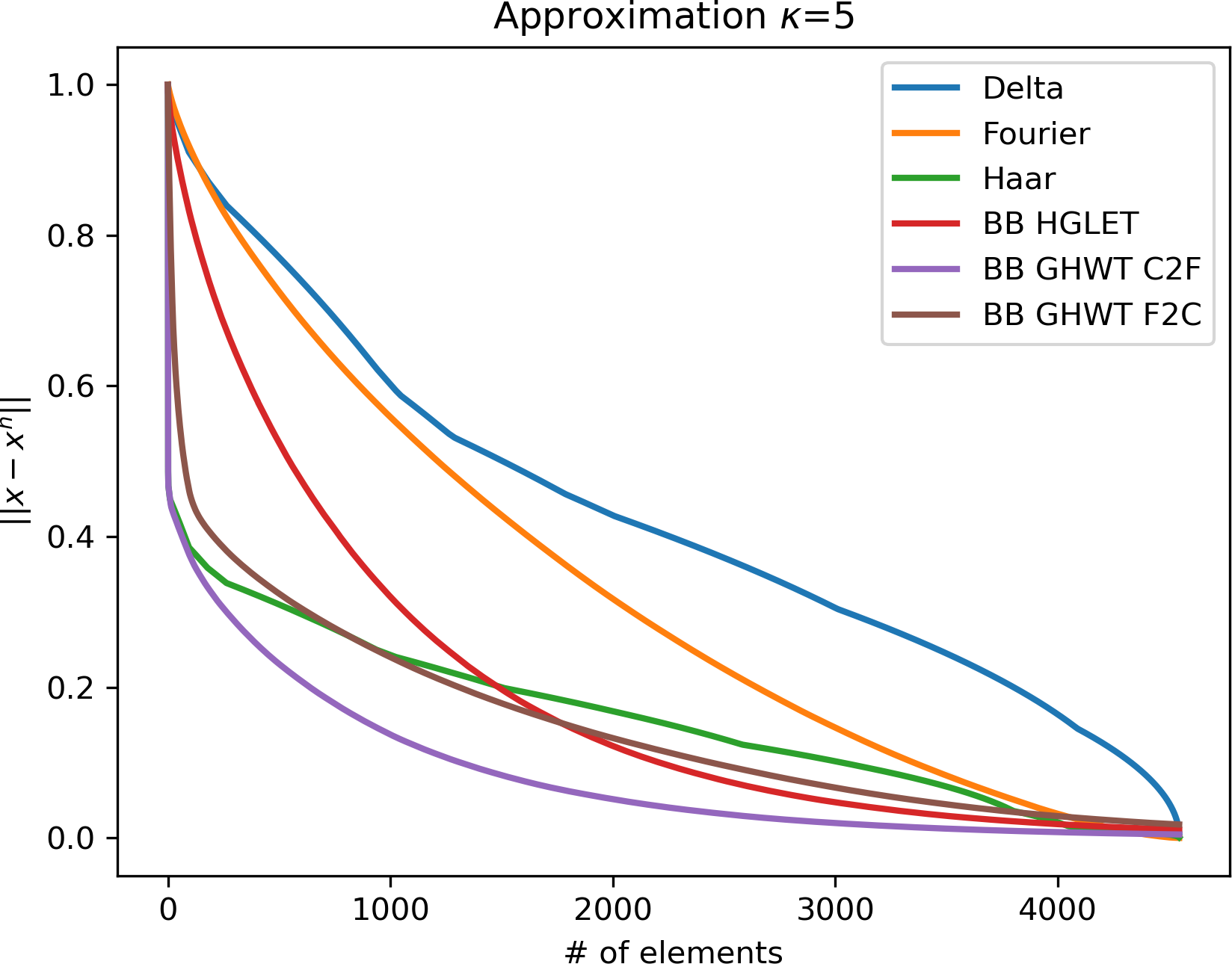}
     \end{subfigure}
     
    \caption{Nonlinear approximation of the Citation Complex for $\kk = 0, \ldots, 5$. Here the HGLET and GHWT bases are selected by the best basis algorithm.}
    \label{fig:CC_approx1}
\end{figure}

\begin{figure}
     \centering
     \begin{subfigure}[b]{0.32\textwidth}
         \centering
         \includegraphics[width=\textwidth]{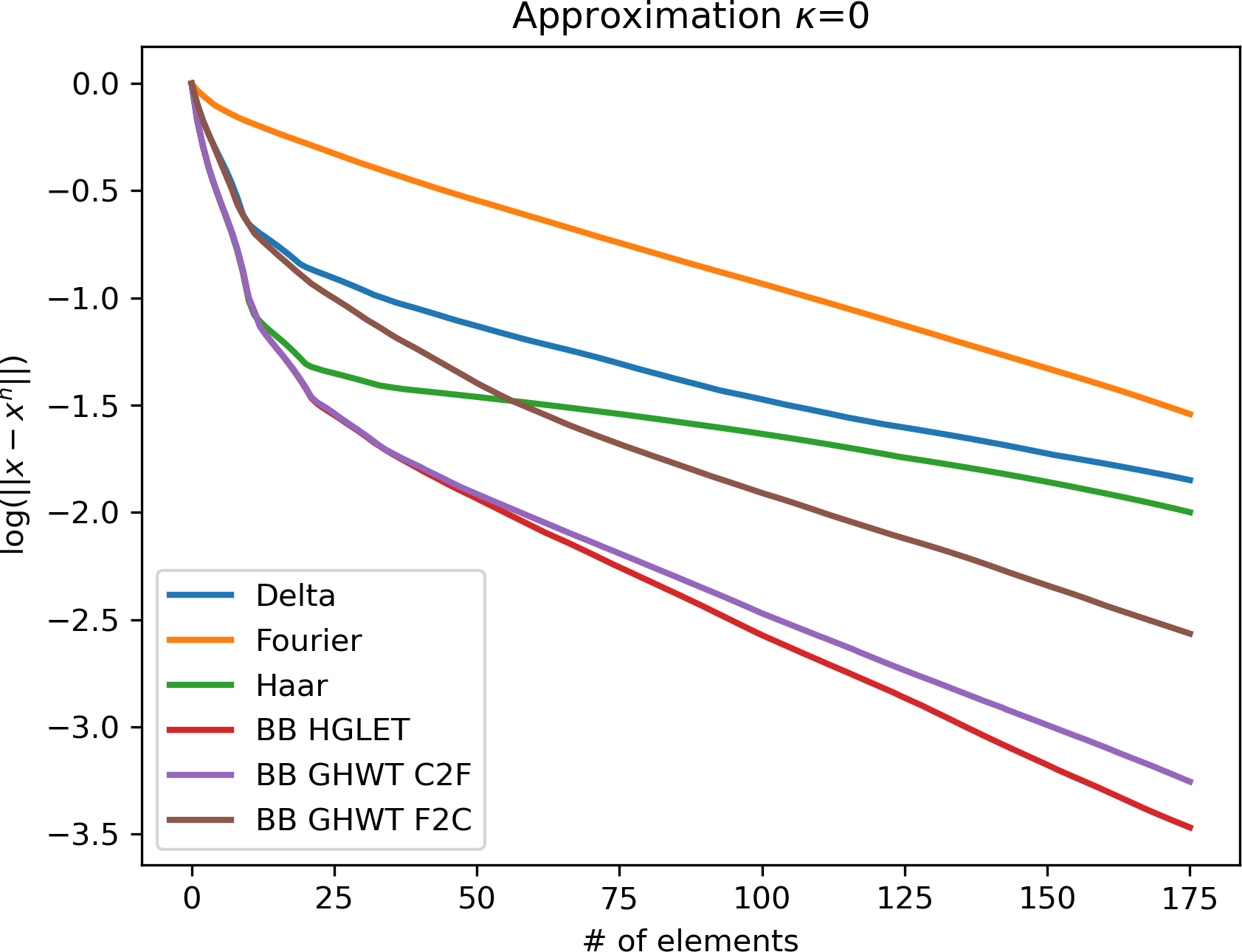}
     \end{subfigure}
     \hfill
     \begin{subfigure}[b]{0.32\textwidth}
         \centering
         \includegraphics[width=\textwidth]{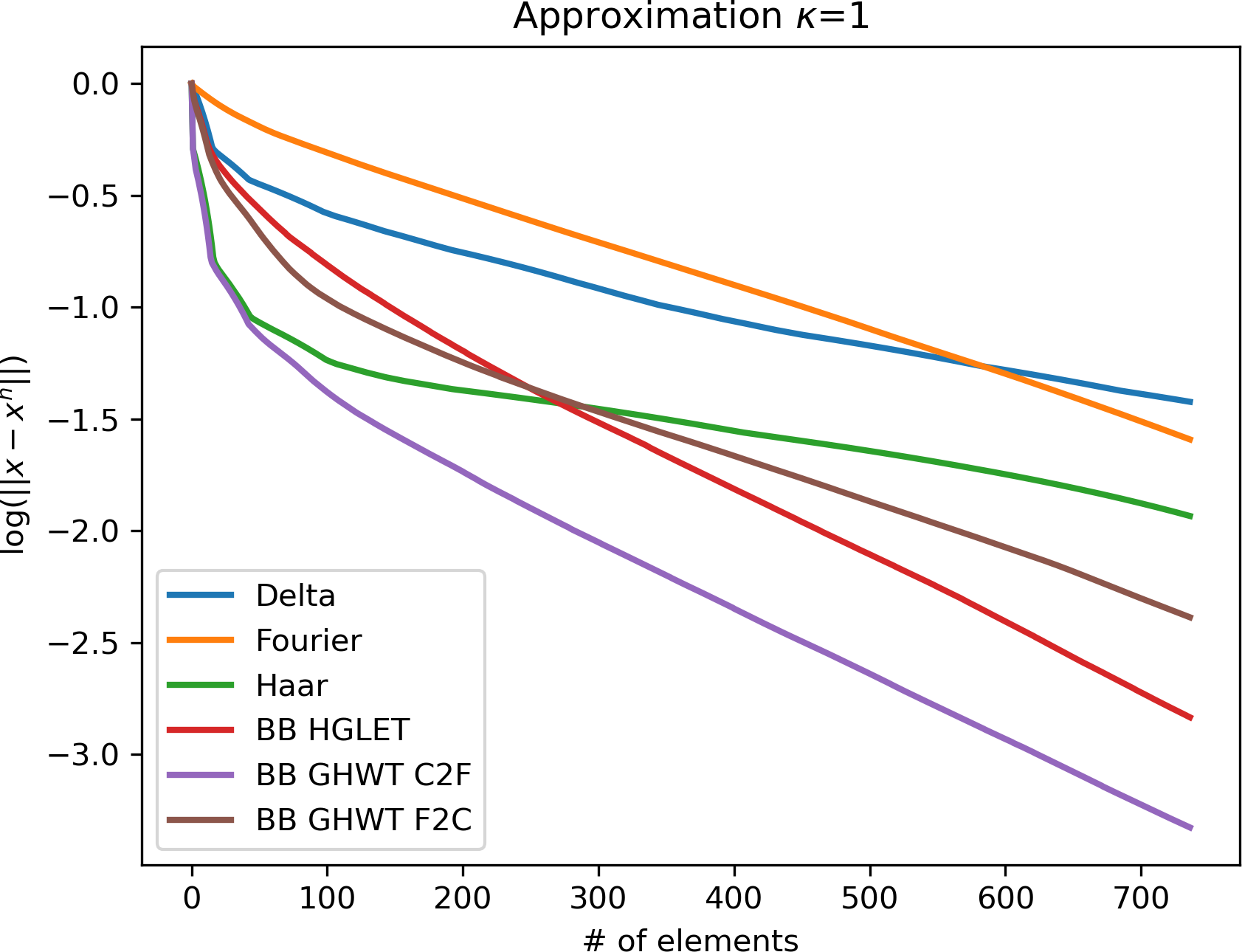}
     \end{subfigure}
     \hfill
     \begin{subfigure}[b]{0.32\textwidth}
         \centering
         \includegraphics[width=\textwidth]{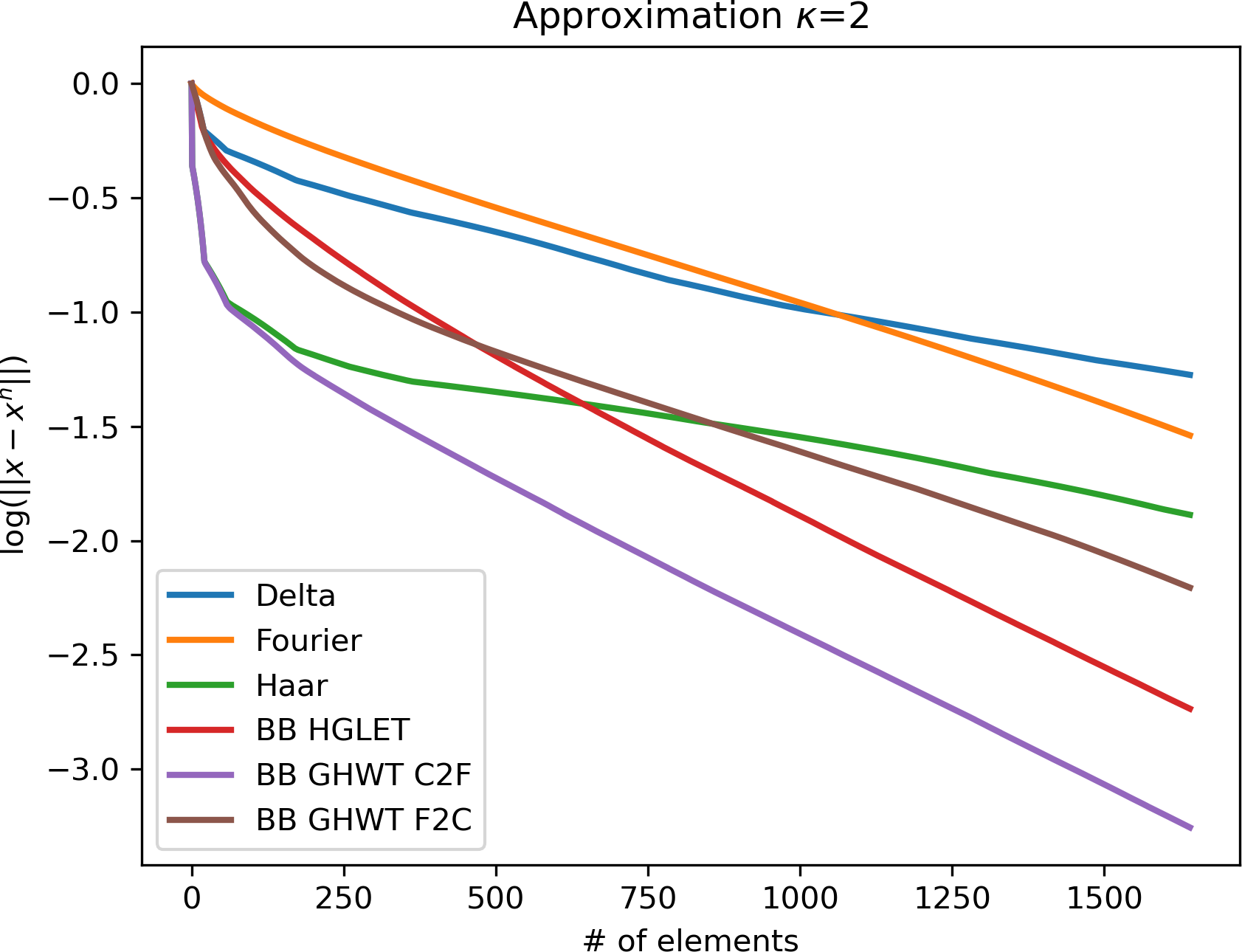}
     \end{subfigure}
     \hfill
     \begin{subfigure}[b]{0.32\textwidth}
         \centering
         \includegraphics[width=\textwidth]{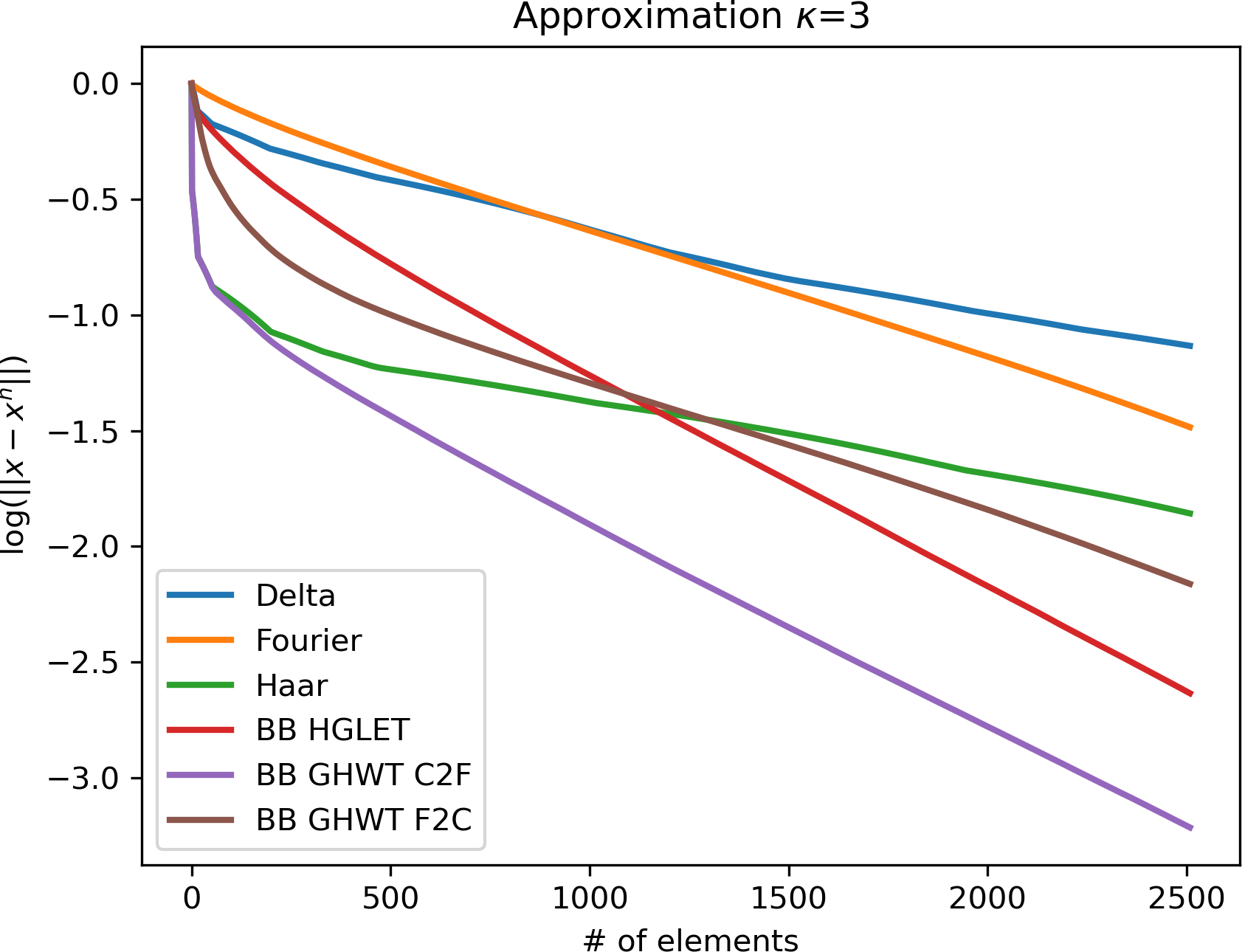}
     \end{subfigure}
     \hfill
     \begin{subfigure}[b]{0.32\textwidth}
         \centering
         \includegraphics[width=\textwidth]{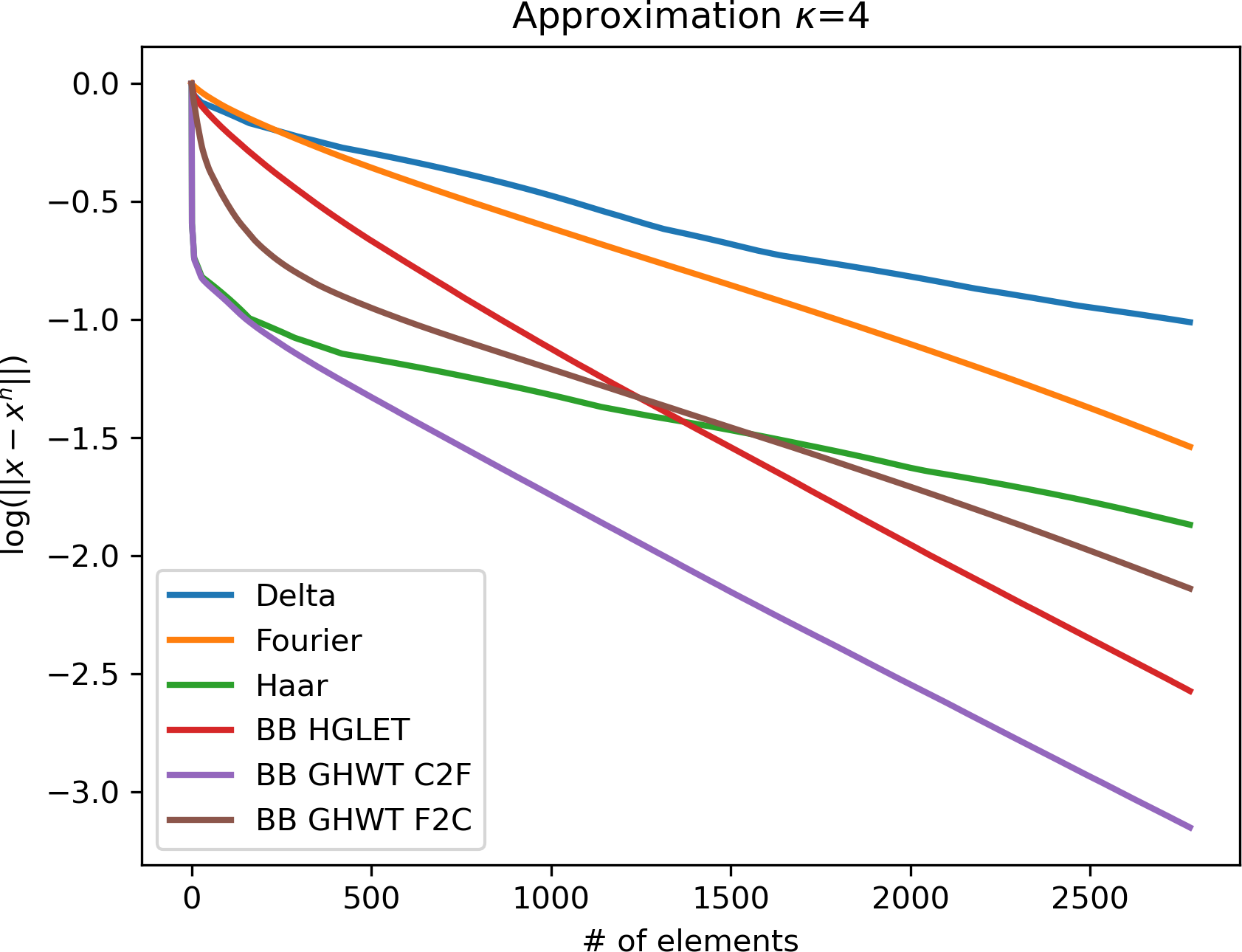}
     \end{subfigure}
     \hfill
     \begin{subfigure}[b]{0.32\textwidth}
         \centering
         \includegraphics[width=\textwidth]{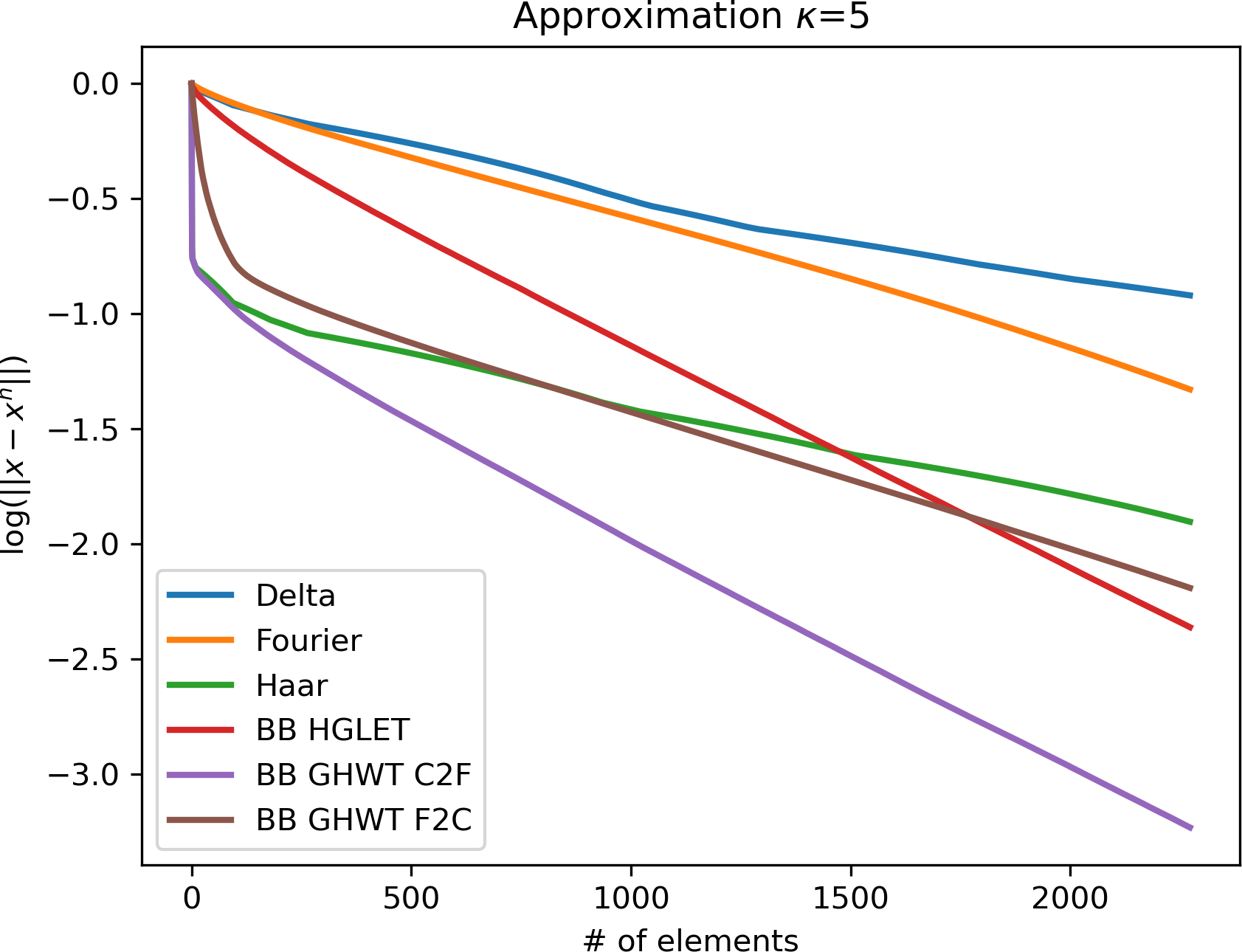}
     \end{subfigure}
    
    \caption{Top: Nonlinear approximation of the Citation Complex for $\kk = 0, \ldots, 5$. Bottom: Log of the error for up to 50\% of the terms retained.}
        \label{fig:CC_appro2x}
\end{figure}

\subsection{Signal Clustering and Classification}
\label{subsec:classification}

Since the basis (and dictionary) vectors we present are both multiscale and built from the Hodge Laplacians that are aware of both topological and geometric properties of the domain \cite{chen2021helmholtzian}, they can function as very powerful feature extractors for general data science applications. In this section, we present two downstream applications: 1) a supervised classification problem; and 2) an unsupervised clustering problem. For baselines, we compare our proposed dictionaries with Fourier and Delta (indicator function) bases and with the Hodgelets proposed in \cite{roddenberry2022hodgelets} for cases when $\kk=1$.

\subsubsection{Supervised Classification}

First, we present our study in supervised classification. We begin by computing edge-valued signals for $1000$ handwritten digits from the MNIST dataset \cite{lecun-mnisthandwrittendigit-2010} by sampling $500$ points in the unit square and following the interpolation method presented for the peppers image in Section~\ref{subsec:approx}. We then compute the features of these images using the proposed orthogonal transforms and best bases from the overcomplete dictionaries. Next, we train a support vector machine (SVM) to classify the digits for each of the transformed representations using the $1000$ training examples. Finally, we test these SVMs on the rest of the whole MNIST dataset. 

We repeat this experiment for the FMNIST dataset \cite{xiao2017online}, again using only $1000$ examples for training data. Results are presented in Table~\ref{table:MNIST}. We remark that these tests are not meant to achieve state-of-the-art results for image classification but rather to showcase the effectiveness of these representations for downstream tasks. Unsurprisingly, the signal-adaptive dictionary methods outperformed the non-adaptive basis methods. Again, the piecewise-constant methods (GHWT, Haar) achieved better approximations than the smoother methods (Fourier, HGLET, Joint, and Separate Hodgelets). This is likely due to the near-binary nature of images in both datasets.

\begin{table}[]
\centering
\resizebox{\textwidth}{!}{
\begin{tabular}{lllllllllllll}
\cline{2-13}
\emph{}   & \multicolumn{7}{c}{\textbf{Basis Methods}}                                                                                                                                                                                                                                                                                                            & \multicolumn{1}{c}{\textbf{}} & \multicolumn{4}{c}{\textbf{Dictionary Methods}}                                                                 \\ \cline{2-8} \cline{10-13} 
\emph{}   & \multicolumn{1}{c}{Delta} & \multicolumn{1}{c}{Fourier} & \multicolumn{1}{c}{Haar} & \multicolumn{1}{c}{Walsh} & \multicolumn{1}{c}{\begin{tabular}[c]{@{}c@{}}HGLET\\ (BB)\end{tabular}} & \multicolumn{1}{c}{\begin{tabular}[c]{@{}c@{}}GHWT\\ (BB C2F)\end{tabular}} & \multicolumn{1}{c}{\begin{tabular}[c]{@{}c@{}}GHWT\\ (BB F2C)\end{tabular}} & \multicolumn{1}{c}{}          & \multicolumn{1}{c}{Joint} & \multicolumn{1}{c}{Separate} & \multicolumn{1}{c}{HGLET} & \multicolumn{1}{c}{GHWT} \\ \cline{2-8} \cline{10-13} 
\# of terms & 661                       & 661                         & 661                      & 661                       & 661                                                                      & 661                                                                         & 661                                                                         &                               & 5288                      & 5288                         & 9254                     & 9254                      \\ \cline{1-8} \cline{10-13} 
MNIST      & 68.675                    & 77.053                      &  75.388         & 77.011                    & 77.991                                                                   & \textbf{78.779}                                                                      & 77.156                                                                      &                               & 79.202                    & 80.038                       &  80.001                   &     \textbf{81.089}      \\
FMNIST       & 64.370                    & 76.753                      & \textbf 76.779          & 75.230                    & 76.117                                                                   & \textbf{76.991}                                                                      & 76.121                                                                      &                               & 78.761                    & 78.738                       & 79.739                 &    \textbf{80.789}          \\ \hline

\end{tabular}
}
\caption{Test Accuracy for SVMs trained on transforms of MNIST signals interpolated to a random triangulation}
\label{table:MNIST}
\end{table}

\subsubsection{Unsupervised Clustering}
\label{subsec:clustering}

A natural setting for studying signals on 1-simplices $C_1$ is the analysis of trajectories~\cite{chen2021helmholtzian, roddenberry2022signal, roddenberry2022hodgelets}. Of particular interest is the case where the domain has nontrivial topological features. Such is the case of the  Global Drifter Program dataset, which tracks the positions of $334$ buoys dropped into the ocean at various points around the island of Madagascar~\cite{roddenberry2022hodgelets}.

We begin by dividing the dataset into three subsets, train ($\vert X_\mathrm{tr} \vert =176$), test ($\vert X_\mathrm{te}  \vert =83$) and validation ($\vert X_\mathrm{vl} \vert =84$). We then use orthogonal matching pursuit~\cite{cai2011orthogonal} (OMP) to compute the $m$ significant features of the training set. Next, we extract these features for the test set and use them to compute the centroids $\{\c_j\}_{j=1}^d$ for each cluster. To evaluate these clusters $K$-score (i.e., the standard $k$-means objective) on the transformed features of the validation set:
\begin{equation*}
    K\mathrm{-score} := \dfrac{1}{N} \sum_{i=1}^{N} \min_{1 \leq j \leq d}   \| \f(\x_i) - \c_j \|^2 , \quad \x_i \in X_\mathrm{vl},
\end{equation*}
where $N=\vert X_\mathrm{vl} \vert = 84$ and $\f(\cdot)$ represents the feature extraction prescribed by applying OMP to the test set. We repeat this experiment for $m=5,10,15,20,25$ (number of features) and $d=2, \ldots, 7 $ (number of clusters). Figure~\ref{fig:bouys_full} summarizes the results of this test, while Table~\ref{table:bouy_full} shows the full numerical results. In this experiment, the GHWT outperformed all other bases because the trajectories are roughly constant and locally supported. The orthogonal matching pursuit scheme can select elements with the correct support size, and the piecewise-constant nature of the GHWT atoms can capture the action of the trajectory with very few elements.

\begin{figure}
    \centering
    \includegraphics[width=.95 \linewidth]{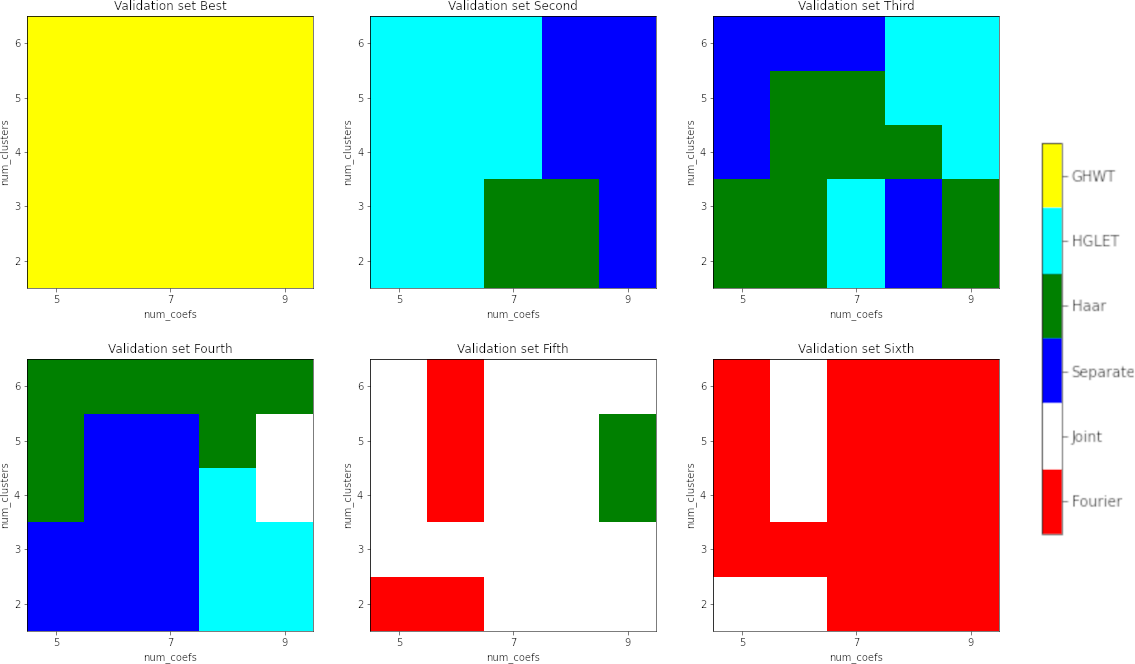}
    \caption{Extensive results for buoy cluster test. Leftmost figure shows which method preformed best, the second to the left shows the second best and so on. The $x$-axis in each subplot indicates the number of coefficients used and the $y$-axis is the number of clusters. Full numerical results are presented in Table~\ref{table:bouy_full}.}
    \label{fig:bouys_full}
\end{figure}

\section{Conclusions and Future Work}
\label{sec:concl}
In this article, we have developed several generalizations of orthonormal bases and overcomplete transforms/dictionaries for signals defined on $\kk$-simplices, and demonstrated their usefulness for data representation on both illustrative synthetic examples and real-world simplicial complexes generated from a co-authorship/citation dataset and an ocean current/flow dataset.
However, there are many more tools from harmonic analysis that we have not addressed in this article. From a theoretical standpoint, future work may involve: 1) defining additional families of multiscale transforms such as the \emph{extended Generalized Haar-Walsh Transform} (eGHWT)~\cite{shao2019extended} and \emph{Natural Graph Wavelet Packets} (NGWPs)~\cite{CLONINGER-LI-SAITO}; 2) exploring different best-basis selection criteria tailored for classification and regression problems such as the \emph{Local Discriminant Basis}~\cite{SAITO-COIF-JMIV, SAITO-COIF-GESHWIND-WARNER} and the \emph{Local Regression Basis}~\cite{SAITO-COIF-SONIC} on simplicial complexes; and 3) investigating nonlinear feature extraction techniques such as the \emph{Geometric Scattering Transform}~\cite{gao2019geometric}. From an application standpoint, we look forward to applying the techniques presented here to data science problems in computational chemistry, weather forecasting, and genetic analysis, all of which have elements that are naturally modeled with simplicial complexes.  

\paragraph{Acknowledgments}
This research was partially supported by the US National Science Foundation
grants DMS-1418779, DMS-1912747, CCF-1934568; the US Office of
Naval Research grant N00014-20-1-2381.

\bibliography{main.bib}

\newpage 

\section{Appendix: Full Results for Buoy Clustering}

\begin{table}[H]
\centering
\resizebox{\textwidth}{!}{
\begin{tabular}{cccccccc}
\hline
Clusters & \# Feat. & Fourier & Joint & Separate & Haar  & HGLET & GHWT           \\ \hline
         & 5        & 0.174   & 0.183 & 0.122    & 0.115 & 0.154 & \textbf{0.024} \\
         & 10       & 0.150   & 0.151 & 0.109    & 0.110 & 0.124 & \textbf{0.023} \\
2        & 15       & 0.129   & 0.129 & 0.120    & 0.093 & 0.119 & \textbf{0.021} \\
         & 20       & 0.118   & 0.113 & 0.108    & 0.084 & 0.107 & \textbf{0.023} \\
         & 25       & 0.104   & 0.099 & 0.096    & 03073 & 0.103 & \textbf{0.024} \\ \hline
         & 5        & 0.174   & 0.163 & 0.110    & .0115 & 0.126 & \textbf{0.026} \\
         & 10       & 0.143   & 0.137 & 0.100    & 0.108 & 0.103 & \textbf{0.023} \\
3        & 15       & 0.126   & 0.112 & 0.113    & 0.095 & 0.118 & \textbf{0.021} \\
         & 20       & 0.114   & 0.104 & 0.100    & 0.081 & 0.095 & \textbf{0.019} \\
         & 25       & 0.099   & 0.092 & 0.089    & 0.069 & 0.093 & \textbf{0.021} \\ \hline
         & 5        & 0.139   & 0.135 & 0.096    & 0.091 & 0.101 & \textbf{0.023} \\
         & 10       & 0.137   & 0.120 & 0.090    & 0.096 & 0.082 & \textbf{0.019} \\
4        & 15       & 0.116   & 0.099 & 0.083    & 0.079 & 0.097 & \textbf{0.018} \\
         & 20       & 0.111   & 0.094 & 0.084    & 0.072 & 0.090 & \textbf{0.021} \\
         & 25       & 0.094   & 0.083 & 0.076    & 0.062 & 0.087 & \textbf{0.022} \\ \hline
         & 5        & 0.135   & 0.116 & 0.087    & 0.081 & 0.074 & \textbf{0.014} \\
         & 10       & 0.118   & 0.109 & 0.083    & 0.090 & 0.062 & \textbf{0.018} \\
5        & 15       & 0.110   & 0.090 & 0.078    & 0.074 & 0.083 & \textbf{0.017} \\
         & 20       & 0.103   & 0.090 & 0.075    & 0.068 & 0.079 & \textbf{0.020} \\
         & 25       & 0.083   & 0.079 & 0.069    & 0.058 & 0.083 & \textbf{0.019} \\ \hline
         & 5        & 0.135   & 0.116 & 0.087    & 0.081 & 0.074 & \textbf{0.014} \\
         & 10       & 0.118   & 0.109 & 0.083    & 0.090 & 0.062 & \textbf{0.018} \\
6        & 15       & 0.110   & 0.090 & 0.078    & 0.074 & 0.083 & \textbf{0.017} \\
         & 20       & 0.103   & 0.092 & 0.075    & 0.068 & 0.073 & \textbf{0.020} \\
         & 25       & 0.083   & 0.073 & 0.069    & 0.058 & 0.083 & \textbf{0.019} \\ \hline
         & 5        & 0.116   & 0.137 & 0.084    & 0.082 & 0.065 & \textbf{0.014} \\
         & 10       & 0.115   & 0.106 & 0.089    & 0.092 & 0.055 & \textbf{0.013} \\
7        & 15       & 0.097   & 0.088 & 0.069    & 0.074 & 0.067 & \textbf{0.013} \\
         & 20       & 0.095   & 0.080 & 0.055    & 0.068 & 0.067 & \textbf{0.014} \\
         & 25       & 0.087   & 0.070 & 0.051    & 0.058 & 0.076 & \textbf{0.013} \\ \hline
\end{tabular}
}
\caption{$K$-score for buoys tests, smaller is better}
\label{table:bouy_full}
\end{table}

\end{document}